\documentclass{llncs}

\newcommand{\nocomments}[0]{0} \newcommand{\noappendix}[0]{0} 

\usepackage[T1]{fontenc}
\usepackage{amsmath}
\usepackage{amssymb}
\usepackage{dsfont}

\usepackage{xspace}
\usepackage{mathtools}
\usepackage{makecell}
\usepackage{physics}

\usepackage{environ}

\newcommand{\repeattheorem}[1]{\begingroup
  \renewcommand{\thetheorem}{\ref{#1}}\expandafter\expandafter\expandafter\theorem
  \csname reptheorem@#1\endcsname
  \endtheorem
  \endgroup
}

\NewEnviron{reptheorem}[1]{\global\expandafter\xdef\csname reptheorem@#1\endcsname{\unexpanded\expandafter{\BODY}}\expandafter\theorem\BODY\unskip\label{#1}\endtheorem
}

\usepackage{mathpartir}
\usepackage{graphicx}
\usepackage[hidelinks]{hyperref}
\usepackage[capitalise]{cleveref}

\crefname{definition}{Def.}{Defs.}
\crefname{proposition}{Prop.}{Props.}
\crefname{algorithm}{Alg.}{Algs.}
\crefname{table}{Tab.}{Tabs.}
\crefname{figure}{Fig.}{Figs.}
\crefname{example}{Ex.}{Exs.}
\crefname{section}{Sec.}{Secs.}
\usepackage[plain]{algorithm}

\usepackage{listings}
\usepackage{lstautogobble}  \usepackage{color}          \definecolor{bluekeywords}{rgb}{0.13, 0.13, 1}
\definecolor{greencomments}{rgb}{0, 0.5, 0}
\definecolor{redstrings}{rgb}{0.9, 0, 0}
\definecolor{graynumbers}{rgb}{0.5, 0.5, 0.5}
\usepackage[noend]{algpseudocode}
\algdef{SE}[DOWHILE]{Do}{doWhile}{\algorithmicdo \ \{ }[1]{\} \algorithmicwhile\ #1}\algnewcommand\algorithmicforeach{\textbf{for each}}
\algdef{S}[FOR]{ForEach}[1]{\algorithmicforeach\ #1\ \algorithmicdo}
\algnewcommand{\LineComment}[1]{\Statex \(\triangleright\) #1}
\algrenewcommand\alglinenumber[1]{\scriptsize #1:}
\algrenewcommand\algorithmicindent{2mm}

\newcommand\keywordfont{\sffamily\bfseries}
\newcommand\kwf{\keywordfont}
\algrenewcommand\algorithmicend{{\keywordfont end}}
\algrenewcommand\algorithmicdo{{\keywordfont do}}
\algrenewcommand\algorithmicwhile{{\keywordfont while}}
\algrenewcommand\algorithmicfor{{\keywordfont for}}
\algrenewcommand\algorithmicforall{{\keywordfont for all}}
\algrenewcommand\algorithmicforeach{{\keywordfont for each}}
\algrenewcommand\algorithmicloop{{\keywordfont loop}}
\algrenewcommand\algorithmicrepeat{{\keywordfont repeat}}
\algrenewcommand\algorithmicuntil{{\keywordfont until}}
\algrenewcommand\algorithmicprocedure{{\keywordfont procedure}}
\algrenewcommand\algorithmicfunction{{\keywordfont function}}
\algrenewcommand\algorithmicif{{\keywordfont if}}
\algrenewcommand\algorithmicthen{{\keywordfont then}}
\algrenewcommand\algorithmicelse{{\keywordfont else}}
\algrenewcommand\algorithmicrequire{{\keywordfont Require:}}
\algrenewcommand\algorithmicensure{{\keywordfont Ensure:}}
\algrenewcommand\algorithmicreturn{{\keywordfont ret}}

\algnewcommand\algorithmicis{{\keywordfont is}}

\usepackage[dvipsnames, table]{xcolor}
\usepackage{booktabs}
\usepackage{multirow}

\usepackage{wrapfig}

\newcommand{\ag}[1]{\textcolor{CarnationPink}{\textbf{AG:} #1}}
\newcommand{\ig}[1]{\textcolor{Emerald}{\textbf{IG:} #1}}
\newcommand{\hg}[1]{\textcolor{Brown}{\textbf{HG:} #1}}
\newcommand{\sharon}[1]{\textcolor{Purple}{\textbf{SS:} #1}}
\newcommand{\todo}[1]{\textcolor{Red}{\textbf{TODO:} #1}}

\ifthenelse{\nocomments = 1}
{\renewcommand{\ag}[1]{}
\renewcommand{\ig}[1]{}
\renewcommand{\hg}[1]{}
\renewcommand{\sharon}[1]{}
\renewcommand{\todo}[1]{}
}{}

\usepackage{tikz}
\usetikzlibrary{shapes.geometric}
\usepackage{subcaption}

\newcommand{\tgnodes}[0]{\ensuremath{\mathit{N}}\xspace}
\newcommand{\tgedges}[0]{\ensuremath{\mathit{E}}\xspace}

\newcommand{\tgreps}[0]{\ensuremath{\mathtt{repr}}\xspace}

\newcommand{\tgroot}[0]{\ensuremath{\mathtt{root}}\xspace}
\newcommand{\tglabel}[0]{\ensuremath{L}\xspace}
\newcommand{\classof}[0]{\ensuremath{\mathit{class}}\xspace}
\newcommand{\toexpr}[0]{\ensuremath{\mathtt{ntt}}\xspace}
\newcommand{\toformula}[0]{\ensuremath{\mathtt{to\_formula}}\xspace}

\newcommand{\egraph}[0]{\ensuremath{\mathit{egraph}}\xspace}
\newcommand{\formula}[0]{\ensuremath{\mathit{isFormula}}\xspace}
\newcommand{\ntt}[0]{\ensuremath{\mathit{term}}\xspace}

\newcommand{\nttrepr}[1]{\ensuremath{\toexpr(#1)}\xspace}
\newcommand{\frontiers}[0]{\ensuremath{\mathcal{F}\xspace}}
\newcommand{\cfrontiers}[0]{\ensuremath{c\mathcal{F}\xspace}}
\newcommand{\egtuple}[0]{\ensuremath{G = \langle \tgnodes, \tgedges, \tglabel, \tgroot
\rangle}}

\newcommand{\limp}{\Rightarrow}
\renewcommand{\Vec}[1]{\vb*{#1}}

\newcommand{\degree}{\ensuremath{\mathtt{deg}}}
\newcommand{\children}{\ensuremath{\mathtt{children}}}

\newcommand{\qel}{QEL\xspace}

\usepackage{mathpartir}
\usepackage{prftree}

\newcommand{\qelim}{qelim\xspace}
\newcommand{\eqdef}{\triangleq}

\newcommand{\vars}{\ensuremath{\Vec{v}}\xspace}
\newcommand{\varsp}{\ensuremath{\Vec{u}}\xspace}

\newcommand{\node}{\mathsf{N}}
\newcommand{\undefrepr}{\bigstar}
\newcommand{\frnt}{\Pi}
\newcommand{\cground}{c\text{-}ground\xspace}

\newcommand{\processQ}{\ensuremath{\mathtt{process}}\xspace}

\newcommand{\checkr}{\ensuremath{\mathtt{match}}\xspace}
\newcommand{\applyr}{\ensuremath{\mathtt{apply}}\xspace}
\newcommand{\egraphmbp}{\textsc{MBP-QEL}\xspace}

\DeclareMathOperator{\ueq}{\mathit{eq}}
\DeclareMathOperator{\eq}{\approx}
\DeclareMathOperator{\deq}{\ensuremath{\not\approx}}

\newcommand{\finddefs}[0]{\ensuremath{\mathtt{find\_defs}}\xspace}
\newcommand{\refinedefs}[0]{\ensuremath{\mathtt{refine\_defs}}\xspace}
\newcommand{\findcore}[0]{\ensuremath{\mathtt{find\_core}}\xspace}
\newcommand{\core}[0]{\ensuremath{\mathtt{core}}\xspace}

\newcommand{\repf}[0]{representative function\xspace}
\newcommand{\filter}{\ensuremath{S}\xspace}

\newcommand{\acq}{\ensuremath{\mathit{todo}}\xspace}
\newcommand{\ztg}{\textsc{Z3eg}\xspace}
\newcommand{\zthree}{\textsc{Z3}\xspace}
\newcommand{\zvanilla}{\textsc{Z3}\xspace}
\newcommand{\eld}{\textsc{Eldarica}\xspace}
\newcommand{\qelite}{\textsc{QeLite}\xspace}
\newcommand{\spacer}{\textsc{Spacer}\xspace}
\newcommand{\notcore}{S}

\newcommand{\fst}{\mathit{fst}}
\newcommand{\snd}{\mathit{snd}}
\newcommand{\ap}{\vec{p}}
\newcommand{\intsort}{I}
\newcommand{\pairsort}{P}
\newcommand{\pair}{\mathit{pair}}
\newcommand{\arread}{\mathit{read}}
\newcommand{\arwrite}{\mathit{write}}
\newcommand{\intarr}{A_{\intsort\times\intsort}}

\newcommand{\condappendix}[1]{\ifthenelse{\noappendix = 1}{}{#1}}

\begin{document}

\title{Fast Approximations of Quantifier Elimination
}
\author{
   Isabel Garcia-Contreras\inst{1}\orcidID{0000-0001-6098-3895}
   \and
   Hari Govind V~K\inst{1}\orcidID{0000-0002-2789-5997}
   \and
   Sharon Shoham\inst{2}\orcidID{0000-0002-7226-3526}
   \and
   Arie Gurfinkel\inst{1}\orcidID{0000-0002-5964-6792}
 }
\institute{University of Waterloo, Waterloo, Canada \\
   \email{\{igarciac,hgvedira,agurfink\}@uwaterloo.ca}
   \and Tel-Aviv University, Tel Aviv, Israel \\ \email{sharon.shoham@cs.tau.ac.il}}
\maketitle              \begin{abstract}

Quantifier elimination (qelim) is used in many automated reasoning tasks including
program synthesis, exist-forall solving, quantified SMT, Model
Checking, and solving Constrained Horn Clauses (CHCs). Exact qelim is computationally expensive. Hence, it is often approximated. For example, Z3 uses ``light'' pre-processing to reduce the number of quantified variables. CHC-solver Spacer uses model-based projection (MBP) to under-approximate qelim relative to a given model, and over-approximations of qelim can be used as abstractions.

In this paper, we present the \qel{} framework for fast approximations of qelim. \qel provides a uniform interface for both quantifier reduction and model-based projection. \qel builds on the egraph data structure -- the core of the EUF decision procedure in SMT -- by casting quantifier reduction as a problem of choosing \emph{ground} (i.e., variable-free) representatives for equivalence classes.
We have used \qel to implement MBP for the theories of Arrays and Algebraic Data Types (ADTs).  We integrated \qel and our new MBP in Z3 and evaluated it within several
tasks that rely on quantifier approximations, outperforming state-of-the-art.

 \end{abstract}

\section{Introduction}

Quantifier Elimination (\qelim) is used in many automated reasoning tasks including program synthesis~\cite{DBLP:conf/pldi/KuncakMPS10}, exist-forall solving~\cite{DBLP:conf/cav/Dutertre14,dutertre2015solving}, quantified SMT~\cite{DBLP:conf/lpar/BjornerJ15}, and Model Checking~\cite{DBLP:conf/cav/KomuravelliGC14}.
Complete \qelim, even when possible, is computationally expensive, and solvers often approximate it. We call these approximations \emph{quantifier reductions}, to separate them from qelim. The difference is that quantifier reduction might leave some free variables in the formula.

For example, Z3~\cite{z3} performs quantifier reduction, called \qelite, by greedily substituting variables by definitions syntactically appearing in the formulas. While it is very useful, it is necessarily  sensitive to the order in which variables are substituted and depends on definitions appearing explicitly in the formula.
Even though it may seem that these shortcomings need to be tolerated to keep \qelite fast, in this paper we show that it is not actually the case; we propose an egraph-based algorithm, \qel, to perform fast quantifier reduction 
that is complete relative to some semantic properties of the formula.

Egraph~\cite{DBLP:conf/focs/NelsonO77} is a data structure that compactly represents infinitely many terms and their equivalence classes. It was initially proposed as a decision procedure for EUF~\cite{DBLP:conf/focs/NelsonO77} and used for theorem proving (e.g., \textsc{Simplify}~\cite{10.1145/1066100.1066102}). Since then, the applications of egraphs have grown. Egraphs are now used as term rewrite systems in equality saturation~\cite{DBLP:conf/pldi/JoshiNR02,DBLP:journals/corr/abs-1012-1802}, for theory combination in SMT solvers~\cite{DBLP:journals/toplas/NelsonO79,10.1145/1066100.1066102}, and for term abstract domains in Abstract Interpretation~\cite{DBLP:conf/fsttcs/GulwaniTN04,DBLP:conf/vmcai/ChangL05,DBLP:conf/vmcai/GangeNSSS16}.

Using egraphs for rewriting or other formula manipulations (like qelim) requires a special operation, called \emph{extract}, that converts nodes in the egraph back into terms. Term extraction was not considered when egraphs were first designed~\cite{DBLP:conf/focs/NelsonO77}. As far as we know, extraction was first studied in the application of egraphs for compiler optimization.
Specifically, equality saturation~\cite{DBLP:conf/pldi/JoshiNR02,10.1145/1480881.1480915} is an optimization technique over egraphs that consists in populating an egraph with many equivalent terms inferred by applying rules. When the egraph is saturated, i.e., applying the rules has no effect, the equivalent term that is most desired, e.g., smallest in size, is \emph{extracted}. This is a recursive process that extracts each sub-term by choosing one representative among its equivalents. 

Application of egraphs to rewriting have recently resurged driven by the \texttt{egg} library~\cite{egg} and the associated workshop\footnote{\url{https://pldi22.sigplan.org/series/egraphs}.}. In~\cite{egg}, the authors show, once again, the power and versatility of this data structure. Motivated by applications of equality saturation, they provide a generic and efficient framework equipped with term extraction, based on an extensible class analysis. 

Egraphs seem to be the perfect data-structure to address the challenges of quantifier reduction: they allow reasoning about infinitely many equivalent terms and  consider all available variable definitions and orderings at once. 
However, things are not always what they appear. The key to quantifier reduction is finding ground (i.e., variable-free) representatives for equivalence classes with free variables. This goes against existing techniques for term extraction since it requires selecting larger, rather than smaller, terms to be representatives. Selecting representatives carelessly makes term extraction diverge. To our surprise, this problem has not been studied so far.  In fact, \texttt{egg}~\cite{egg} incorrectly claims that any representative function can be used with its term extraction, while the implementation diverges. In this paper, we bridge this gap by providing necessary and sufficient conditions for a representative function to be admissible for term extraction as defined in~\cite{DBLP:conf/pldi/JoshiNR02,egg}. Furthermore, we extend extraction from terms to formulas to enable extracting a formula of the egraph.

Our main contribution is a new quantifier reduction algorithm, called \qel. 
Building on the term extraction described above, it is formulated as finding a representative function that maximizes the number of ground terms as representatives. Furthermore, it greedily attempts to represent variables without ground representatives in terms of other variables, thus further reducing the number of variables in the output. We show that \qel is complete relative to ground definitions entailed by the formula. Specifically, \qel guarantees to eliminate a variable if it is equivalent to a ground term.

Whenever an  application requires eliminating all free variables, incomplete techniques such as \qelite or \qel are insufficient. In this case, qelim is under-approximated using a Model-based Projection (MBP) that uses a model $M$ of a formula to guide under-approximation using equalities and variable definitions that are consistent with $M$.
In this paper, we show that MBP can be implemented using our new techniques for \qel together with the machinery from equality saturation. Just like SMT solvers use egraphs as glue to combine different theory solvers, we use egraphs as glue to combine projection for different theories. In particular, we give an algorithm for MBP in the combined theory of Arrays and Algebraic DataTypes~(ADTs). The algorithm uses insights from \qel to produce less under-approximate MBPs.

We implemented \qel and the new MBP using egraphs inside the state-of-art SMT solver \zthree~\cite{z3}. Our implementation (referred to as \ztg) replaces the existing \qelite and MBP. We evaluate our algorithms in two contexts. 
First, inside the QSAT~\cite{DBLP:conf/lpar/BjornerJ15} algorithm for quantified satisfiability. The performance of QSAT in \ztg is improved, compared to QSAT in \zthree, when ADTs are involved.
Second, we evaluate our algorithms inside the Constrained Horn Clause~(CHC) solver \spacer~\cite{DBLP:conf/cav/KomuravelliGC14}. Our experiments show that \spacer in \ztg solves many more benchmarks containing nested Arrays and ADTs.

\paragraph{Related Work.}
Quantifier reduction by variable substitution is widely used in quantified SMT~\cite{DBLP:conf/fmcad/GasconSDTJM14,DBLP:conf/lpar/BjornerJ15}.
To our knowledge,  we are the first to look at this problem semantically and provide an algorithm that guarantees that the variable is eliminated if the formula entails that it has a ground definition.

Term extraction for egraphs comes from equality saturation~\cite{DBLP:conf/pldi/JoshiNR02,10.1145/1480881.1480915}.
The \texttt{egg} Rust library~\cite{egg} is a recent implementation of equality saturation that supports rewriting and term extraction. However, we did not use \texttt{egg} because we  integrated \qel within Z3 and built it using Z3 data structures instead.

Model-based projection was first introduced for the \spacer CHC solver for LIA and LRA~\cite{DBLP:conf/cav/KomuravelliGC14} and extended to the theory of Arrays~\cite{DBLP:conf/fmcad/KomuravelliBGM15} and ADTs~\cite{DBLP:conf/lpar/BjornerJ15}. Until now, it was implemented by syntactic rewriting. Our egraph-based MBP implementation  is less sensitive to syntax and, more importantly, allows for combining MBPs of multiple theories for MBP of the combination. As a result, our MBP is more general and less model dependent. Specifically, it requires fewer model equalities and produces more general under-approximations than~\cite{DBLP:conf/fmcad/KomuravelliBGM15,DBLP:conf/lpar/BjornerJ15}.  

\medskip
\noindent
\textit{Outline.} The rest of the paper is organized as follows. \cref{sec:background} provides  background. \cref{sec:egraphs-extra} introduces term extraction, extends it to formulas, and characterizes representative-based term extraction for egraphs.
\cref{sec:qel} presents \qel, our algorithm for fast quantifier reduction that is relatively complete.
\cref{sec:mbp} shows how to compute MBP combining equality saturation and the ideas from \cref{sec:qel} for the theories of ADTs and Arrays.
All algorithms have been implemented in \zthree and evaluated in \cref{sec:evaluation}. \condappendix{Proofs are deferred to the appendix.}
 
\section{Background}\label{sec:background}
We assume the reader is familiar with multi-sorted first-order logic (FOL) with
equality and the theory of equality with uninterpreted functions~(EUF) (for an
introduction see, e.g.~\cite{barrett2018satisfiability}).
We use
$\eq$ to denote the designated logical equality symbol. 
For simplicity of presentation, we assume that the FOL signature $\Sigma$ contains
only functions (i.e., no predicates) and constants (i.e., 0-ary functions).
To represent predicates, we assume the FOL signature has a designated sort $\mathsf{Bool}$, and
two $\mathsf{Bool}$ constants $\top$ and $\bot$, representing true, and false
respectively. We then use $\mathsf{Bool}$-valued functions to represent predicates, using
$P(a) \eq \top$ and $P(a) \eq \bot$ to mean that $P(a)$ is true or false,
respectively. Informally, we continue to write $P(a)$ and $\neg P(a)$ as a
syntactic sugar for $P(a) \eq \top$ and $P(a) \eq \bot$, respectively. We use
lowercase letters like $a$, $b$ for constants, and $f$, $g$ for functions, and
uppercase letters like $P$, $Q$ for $\mathsf{Bool}$ functions that represent
predicates.
We denote by $\psi^{\exists}$ the existential closure of $\psi$.

\paragraph{Quantifier Elimination (\qelim).} 
Given a quantifier-free~(QF) formula $\varphi$ with free variables $\Vec{v}$,
\textit{quantifier elimination} of $\varphi^\exists$ is the problem of finding a QF formula $\psi$ with no free variables such that $\psi \equiv \varphi^\exists$.
For example, a \qelim of $\exists a\cdot (a \eq x \land f(a) > 3)$ is $f(x) > 3$; and, 
there is no \qelim of $\exists x \cdot (f(x) > 3)$, because it is 
impossible to restrict $f$ to have ``at least one value in its range that is
greater than 3'' without a quantifier.

\paragraph{Model Based Projection (MBP).} Let $\varphi$ be a formula with free variables $\vars$, and $M$ a model of $\varphi$. A \emph{model-based projection} of $\varphi$ relative to $M$ is a QF formula $\psi$ such that $\psi \limp \varphi^\exists$ and $M \models \psi$. That is, $\psi$ has no free variables, is an under-approximation of $\varphi$, and satisfies the designated model $M$, just like $\varphi$. MBP is used by many algorithms to under-approximate qelim, when the computation of qelim is too expensive or, for some reason, undesirable. 

\paragraph{Egraphs.}
An egraph is a well-known data structure to compactly represent a set of terms
and an equivalence relation on those terms~\cite{DBLP:conf/focs/NelsonO77}. Throughout the paper, we assume that graphs have an ordered successor relation
and use $n[i]$ to denote the $i$th successor (child) of a node $n$. An out-degree
of a node $n$, $\degree(n)$, is the number of edges leaving $n$.
Given a node $n$, $\texttt{parents}(n)$ denotes the set of nodes with an
outgoing edge to $n$ and $\children(n)$ denotes the set of nodes with an
incoming edge from $n$.

\begin{definition} \label{def:egraph2}
  Let $\Sigma$ be a first-order logic signature. An \emph{egraph} is a tuple \allowbreak\mbox{$\egtuple$},~where
\begin{enumerate}
  \item[(a)] $\langle \tgnodes, \tgedges \rangle$ is a directed acyclic graph, 
  \item[(b)]
  \tglabel maps nodes to function symbols in $\Sigma$ or logical variables, and 
  \item[(c)] $\tgroot :
  \tgnodes \mapsto \tgnodes$ maps a node to its \emph{root} such that the relation
  $\rho_{\tgroot} \eqdef \{ (n, n') \mid \tgroot(n) = \tgroot(n') \}$ is an
  equivalence relation on $N$ that is \emph{closed under congruence}: $(n,n') \in \rho_{\tgroot}$ whenever $n$ and $n'$ 
  are \emph{congruent} under $\tgroot$, 
  i.e., whenever $\tglabel(n) =   \tglabel(n')$,  $\degree(n) = \degree(n') > 0$, and, 
  $\forall  1 \leq i \leq \degree(n) \cdot (n[i],n'[i]) \in \rho_{\tgroot}$.
  \end{enumerate}
\end{definition}

Given an egraph $G$, the \emph{class} of a node $n \in G$, $\classof(n) \eqdef
\rho_{\tgroot}(n)$, is the set of all nodes that are equivalent to $n$.
The term
of $n$, $\ntt(n)$, with $\tglabel(n) = f$ is $f$ if $\degree(n) = 0$ and
$f(\ntt(n[1]),\ldots, \ntt(n[\degree(n)]))$, otherwise. We assume that the terms of different nodes are different, and refer to a node $n$ by its term.

\begin{figure}[t]
  \centering
  \scalebox{0.9}{$\varphi_{\ref{fig:egraph-example}}(x,y,z) \eqdef z \eq
  \mathit{read}(a,x) \land k + 1 \eq \mathit{read}(a,y) \land x \eq y \land 3
  > z$}\\
  \rule{7cm}{.2mm}
  \\[-1mm]
  \centering
  \scalebox{0.9}{
  \begin{tikzpicture}[thick, node distance={13mm}, main/.style = {draw, circle, minimum height=7mm}]
\node[main] (1) {$>$};
    \node[main] (21) [left of=1] {$\top$};
    \node[main] (2) [below left of=1,xshift=+4mm] {$3$};
    \node[main] (3) [below right of=1,xshift=-4mm] {$z$};
    \node[main,ellipse] (4) [right of=3] {$\mathit{read}$};
    \node[main] (5) [below left of=4,xshift=+4mm] {$a$};
    \node[main] (6) [below right of=4,xshift=-4mm] {$x$};
    \node[main,ellipse] (7) [right of=4,xshift=+3mm] {$\mathit{read}$};
    \node[main] (8) [right of=6,xshift=-2mm] {$y$};
    \node[main] (9) [right of=7, xshift=+4mm] {$+$};
    \node[main] (10) [below left of=9, xshift=+4mm] {$k$};
    \node[main] (11) [below right of=9, xshift=-4mm] {$1$};

\node [left of=21,xshift=+8mm,yshift=+2mm] {$^{(0)}$};
    \node [left of=1,xshift=+8mm,yshift=+2mm] {$^{(1)}$};
    \node [left of=2,xshift=+8mm,yshift=+2mm] {$^{(2)}$};
    \node [left of=3,xshift=+8mm,yshift=+2mm] {$^{(3)}$};
    \node [right of=4,xshift=-6mm,yshift=+2mm] {$^{(4)}$};
    \node [right of=7,xshift=-6mm,yshift=+2mm] {$^{(5)}$};
    \node [right of=9,xshift=-8mm,yshift=+2mm] {$^{(6)}$};
    \node [left of=5,xshift=+8mm,yshift=+2mm] {$^{(7)}$};
    \node [right of=6,xshift=-8mm,yshift=+2mm] {$^{(8)}$};
    \node [right of=8,xshift=-8mm,yshift=+2mm] {$^{(9)}$};
    \node [right of=10,xshift=-8mm,yshift=+2mm] {$^{(10)}$};
    \node [right of=11,xshift=-8mm,yshift=+2mm] {$^{(11)}$};

\draw[->] (1) -- (2);
    \draw[->] (1) -- (3);
    \draw[->] (4) -- (5);
    \draw[->] (4) -- (6);
    \draw[->] (7) -- (8);
    \draw[->] (7) -- (5);
    \draw[->] (9) -- (10);
    \draw[->] (9) -- (11);

\draw[dashed,->,red] (1) -- (21);
    \draw[dashed,->,red] (8) -- (6);
    \draw[dashed,->,red] (4) -- (3);
    \draw[dashed,->,red] (7) to [out=120, in=30,looseness=1] (3);
    \draw[dashed,->,red] (9) to [out=150, in=60] (3);
  \end{tikzpicture}
}
\caption{Example egraph of $\varphi_{\ref{fig:egraph-example}}$.\label{fig:egraph-example}}
\end{figure}
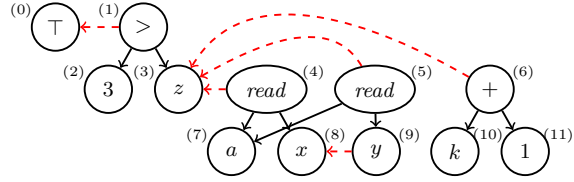

An example of an egraph $\egtuple$ is shown in \cref{fig:egraph-example}. A
symbol $f$ inside a circle depicts a node $n$ with label $\tglabel(n) = f$,
solid black and dashed red arrows depict $\tgedges$ and $\tgroot$, respectively.
The order of the black arrows from left to right defines the order of the
children.
In our examples, we refer to a specific node $i$ by its number using $\node(i)$
or its term, e.g., $\node(k + 1)$.
A node $n$ without an outgoing red arrow is its own root. A set of nodes
connected to the same node with red edges forms an equivalence class.
In this example, $\tgroot$ defines the equivalence classes $\{\node(3), \node(4), \node(5), \node(6)\}$,
$\{\node(8), \node(9)\}$, and a class for each of the remaining nodes. Examples of some terms in $G$ are $\ntt(\node(9)) = y$ and $\ntt(\node(5)) =
\mathit{read}(a,y)$.

\paragraph{An Egraph of a Formula.}
We consider formulas that are conjunctions of equality literals (recall that we represent predicate applications by equality literals). 
Given a formula \mbox{$\varphi~\eqdef~(t_1 \eq u_1 \land \cdots \land t_k \eq u_k)$}, an egraph from $\varphi$ is built (following the standard procedure~\cite{DBLP:conf/focs/NelsonO77}) by creating nodes for each $t_i$ and
$u_i$, recursively creating nodes for their subexpressions, and merging the
classes of each pair $t_i$ and $u_i$, computing the congruence closure for
$\tgroot$.
We write $\egraph(\varphi)$ for an egraph of $\varphi$
constructed via some deterministic procedure based on the recipe above.
\cref{fig:egraph-example} shows an $\egraph(\varphi_1)$ of $\varphi_{\ref{fig:egraph-example}}$.
The equality $z \eq \mathit{read}(a,x)$ is captured by $\node(3)$ and $\node(4)$
belonging to the same class (i.e., red arrow from $\node(4)$ to $\node(3)$).
Similarly, the equality $x \eq y$ is captured by a red arrow from $\node(9)$ to $\node(8)$.
Note that by congruence, $\varphi_{\ref{fig:egraph-example}}$ implies $\mathit{read}(a,x) \eq
\mathit{read}(a,y)$, which, by transitivity, implies that $k + 1 \eq
\mathit{read}(a,x)$. In \cref{fig:egraph-example}, this corresponds to red
arrows from $\node(5)$ and $\node(6)$ to $\node(3)$.
The predicate application $3 > z$ is captured by the red arrow from $\node(1)$
to $\node(0)$. From now on, we omit $\top$ and $\bot$ and the corresponding
edges from figures to avoid clutter.

\begin{figure}[t]
  \centering
  \scalebox{0.9}{$\varphi_{\ref{fig:interpret}}(x,y) \eqdef \ueq(c, f(x)) \land \ueq(d,  f(y)) \land \ueq(x, y)$}\\
  \rule{7cm}{.2mm}
  \\[1mm]
  \centering
  \scalebox{0.9}{
  \begin{subfigure}{0.32\textwidth}
    \begin{tikzpicture}[thick, node distance={10mm}, main/.style = {draw, circle,minimum height=6mm}]
      \node[main] (1) {$c$};
      \node[main] (2) [right of=1] {$f$};
      \node[main] (3) [right of=2] {$d$};
      \node[main] (4) [right of=3] {$f$};
\node[main] (6) [below of=2] {$x$};
      \node[main] (7) [below of=4] {$y$};
      \node[main] () [draw=white,above right of=3,xshift=-2.5mm] {};

      \draw[->] (2) to (6);
      \draw[->] (4) to (7);

      \draw[dashed,->,red] (1) to [out=30, in=120,looseness=.7] (3);
      \draw[dashed,->,red] (2) -- (3);
      \draw[dashed,->,red] (4) -- (3);
\draw[dashed,->,red] (6) -- (7);
    \end{tikzpicture}
    \vfill
    \caption{$G_a$, interpreting $\ueq$ as $\eq$.\label{fig:egraph-eq}}
  \end{subfigure} \hfill
\begin{subfigure}{0.32\textwidth}
    \begin{tikzpicture}[thick, node distance={10mm}, main/.style = {draw, circle,minimum height=6mm}]
        \node[main] (1) {$c$};
        \node[main] (2) [right of=1] {$f$};
        \node[main] (3) [right of=2] {$d$};
        \node[main] (4) [right of=3] {$f$};
\node[main] (6) [below of=2] {$x$};
        \node[main] (7) [below of=4] {$y$};
        \node[main] (8) [above right of=1,xshift=-2.5mm] {$\ueq$};
        \node[main] (9) [above right of=3,xshift=-2.5mm] {$\ueq$};

        \node[main] (11) [below of=3,yshift=+2.5mm] {$\ueq$};

        \draw[->] (2) to (6);
        \draw[->] (4) to (7);

        \draw[->] (8) to (1);
        \draw[->] (8) to (2);

        \draw[->] (9) to (3);
        \draw[->] (9) to (4);

        \draw[->] (11) to (6);
        \draw[->] (11) to (7);
\end{tikzpicture}
    \caption{$G_b$, not interpreting $\ueq$.\label{fig:egraph-ueq}}
  \end{subfigure} \hfill
  \begin{subfigure}{0.34\textwidth}
    \begin{tikzpicture}[thick, node distance={10mm}, main/.style = {draw, circle,minimum height=6mm}]
         \node[main] (1) {$c$};
        \node[main] (2) [right of=1] {$f$};
        \node[main] (3) [right of=2] {$d$};
        \node[main] (4) [right of=3] {$f$};
\node[main] (6) [below of=2] {$x$};
        \node[main] (7) [below of=4] {$y$};
        \node[main] (8) [above right of=1,xshift=-2.5mm] {$\ueq$};
        \node[main] (9) [above right of=3,xshift=-2.5mm] {$\ueq$};

        \node[main] (11) [below of=3,yshift=+2.5mm] {$\ueq$};

        \draw[->] (2) to (6);
        \draw[->] (4) to (7);

        \draw[->] (8) to (1);
        \draw[->] (8) to (2);

        \draw[->] (9) to (3);
        \draw[->] (9) to (4);

        \draw[->] (11) to (6);
        \draw[->] (11) to (7);
        \draw[dashed,->,red] (1) to [out=30, in=120,looseness=.7] (3);
        \draw[dashed,->,red] (2) -- (3);
        \draw[dashed,->,red] (4) -- (3);
\draw[dashed,->,red] (6) to [out=330, in=210,looseness=.7] (7);
    \end{tikzpicture}
    \caption{$G_c$, combining $(a)$ and $(b)$.\label{fig:egraph-eq-ueq}}
  \end{subfigure}
  }
  \vspace{-2mm}
  \caption{Different egraph interpretations for $\varphi_{\ref{fig:interpret}}$.\label{fig:interpret}}
\end{figure}
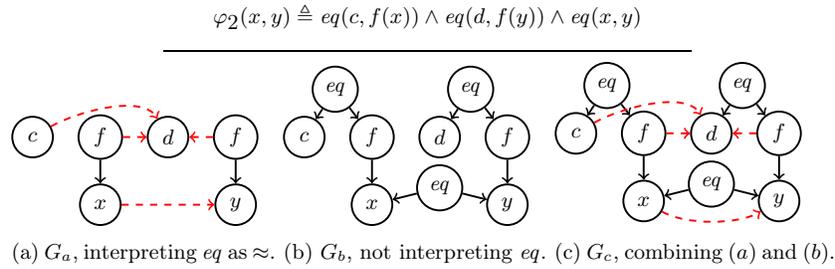

\vspace*{-2mm}
\paragraph{Explicit and Implicit Equality.} Note that egraphs represent equality
implicitly by placing nodes with equal terms in the same equivalence class. Sometimes,
it is necessary to represent equality explicitly, for example, when using
egraphs for equality-aware rewriting (e.g., in \texttt{egg}~\cite{egg}). To represent equality
explicitly, we introduce a binary $\mathsf{Bool}$ function $\ueq$ and write $\ueq(a, b)$
for an equality that has to be represented explicitly. We change the $\egraph$
algorithm to treat $\ueq(a, b)$ as both a function application, and as a logical
equality $a \eq b$: when processing term $\ueq(a, b)$, the algorithm
both adds $\ueq(a, b)$ to the egraph, and merges the nodes for $a$ and $b$ into one class.
For example, \cref{fig:interpret} shows three different interpretations of a formula
$\varphi_{\ref{fig:interpret}}$ with equality interpreted: implicitly (as
in~\cite{DBLP:conf/focs/NelsonO77}), explicitly (as in~\cite{egg}), and both
implicitly and explicitly (as in this paper).

\section{Extracting Formulas from Egraphs}
\label{sec:egraphs-extra}

Egraphs were proposed as a decision procedure for EUF~\cite{DBLP:conf/focs/NelsonO77} -- a setting in which converting an egraph back to a formula, or  \emph{extracting}, is irrelevant. Term extraction has been studied in the context of equality saturation and term rewriting~\cite{DBLP:conf/pldi/JoshiNR02,egg}. However, existing literature presents extraction as a heuristic, and, to the best of our knowledge, has not been exhaustively explored.
In this section, we fill these gaps in the literature and extend extraction from terms to formulas.

\paragraph{Term Extraction.} We begin by recalling how to extract the term of a node. The function $\toexpr$ (node-to-term) in \cref{alg:eg-to-cube} does an extraction parametrized by a representative function $\tgreps : N \mapsto N$ (same as in~\cite{egg}). A function \tgreps assigns each class a unique representative node (i.e., nodes in the same class are mapped to the same representative) so that  $\rho_{\tgroot} = \rho_{\tgreps}$. The function $\toexpr$ extracts a term of a node recursively, similarly to $\ntt$, except that the representatives of the children of a node are used instead of the actual children. We refer to terms built in this way by $\nttrepr{n, \tgreps}$ and omit \tgreps when it is clear~from~the~context.

As an example, consider $\tgreps_1 \eqdef \{\node(3), \node(8))\}$ for \cref{fig:egraph-example}. For readability, we denote representative functions by sets of nodes that are the class representatives, omitting $\node(\top)$ that always represents its class, and omitting all singleton classes. Thus, $\tgreps_1$ maps all nodes in $\classof(\node(3))$ to $\node(3)$, nodes in $\classof(\node(8))$ to $\node(8)$, nodes in $\classof(\node(\top))$ to $\node(\top)$, and all singleton classes to themselves. For example, $\toexpr(\node(5))$ extracts $\mathit{read}(a,x)$, since $\node(9)$ has as representative $\node(8)$.

\newcommand{\lits}{\ensuremath{\mathit{Lits}}}

\begin{figure}[t]
  \small
  \centering
  \scalebox{0.9}{
  \begin{minipage}[t]{0.5\textwidth}
    $\egraph::\toformula(\tgreps,\filter)$
    \begin{algorithmic}[1]
    \State $\lits := \emptyset$
    \For{$r = \tgreps(r) \in N$}\label{ln:opfmlz3:nodes}
    \State $t := \toexpr(r,\tgreps)$
    \For {$n \in (\classof(r) \setminus r)$}
    \If{$n \not\in \filter$}
    \State $\lits := \lits \cup \{t \eq \toexpr(n,\tgreps)\}$ \label{l:eqs}
    \EndIf
    \EndFor
    \EndFor
    \State \Return $\bigwedge \lits$
\algstore{end-toformula}
  \end{algorithmic}
\end{minipage}\hspace{0.5in}
\begin{minipage}[t]{0.5\textwidth}
  $\egraph::\toexpr(n,\tgreps)$
  \begin{algorithmic}[1]
    \algrestore{end-toformula}
    \State $f := \tglabel[n]$
    \If{$\degree(n) = 0$}
    \State \Return $f$
    \Else
    \For{$i \in [1, \degree(n)]$}
    \State $\mathit{Args}[i] := \toexpr(\tgreps(n[i]), \tgreps)$\label{ln:opfmlz3:reprchildren}
    \EndFor
    \State \Return $f(\mathit{Args})$
    \EndIf
  \end{algorithmic}
\end{minipage}
  }
  \vspace{-2mm}
\caption{Producing formulas from an egraph.\label{alg:eg-to-cube}}
\end{figure}

\paragraph{Formula Extraction.} Let $G = \egraph(\varphi)$ be an egraph of some formula $\varphi$. A formula $\psi$
is a \emph{formula of} $G$, written $\formula(G,\psi)$, if $\psi^\exists \equiv \varphi^\exists$. 

\Cref{alg:eg-to-cube} shows an algorithm $\toformula(\tgreps, \filter)$
to compute a formula $\psi$ that satisfies $\formula(G, \psi)$ for a given egraph $G$.
In addition to $\tgreps$, $\toformula$ is parameterized by a set of nodes $\filter \subseteq \tgnodes$ to exclude\footnote{The set $\filter$ affects the result, but for this section, we restrict to the case of $\filter \eqdef \emptyset$.}.
To produce the equalities corresponding to the classes, for each representative $r$, for each
$n \in (\classof(r) \setminus \{r\})$ the output formula has a literal $\nttrepr{r}
\eq \nttrepr{n}$. 
For example, using $\tgreps_1$ for the egraph in \cref{fig:egraph-example}, we obtain
for $\classof(\node(8))$, $(x \eq y)$; for $\classof(\node(3))$, $(z \eq \mathit{read}(a,x) \land z \eq \mathit{read}(a,x) \land z \eq k + 1)$; and for $\classof(\node(0))$, $(\top \eq 3 > z)$.
The final result (slightly simplified) is:
$ x \eq y \land z \eq \mathit{read}(a,x) \land z \eq k + 1 \land 3 > z$.

Let $G = \egraph(\varphi)$ for some formula $\varphi$. Note that, $\psi$ computed by $\toformula$ is not syntactically the same as $\varphi$. That is, $\toformula$ is not an inverse of $\egraph$. Furthermore, since $\toformula$
commits to one representative per class, it is limited in what formulas it can generate. For example, since $x \eq y$ is in $\varphi_{\ref{fig:egraph-example}}$, for any
$\tgreps$, $\varphi_{\ref{fig:egraph-example}}$ cannot be the result of $\toformula$, because the output can contain only one of $\mathit{read}(a,x)$ or $\mathit{read}(a,y)$.

\paragraph{Representative Functions.} 
The representative function is instrumental for determining the terms that appear in the extracted formula. To illustrate the importance of representative choice, consider the formula $\varphi_{\ref{fig:tgreps-ex}}$ of \cref{fig:tgreps-ex} and its egraph $G_4 = \egraph(\varphi_{\ref{fig:tgreps-ex}})$. For now, ignore the blue dotted lines. For $\tgreps_{\ref{fig:tgreps-exa}}$, \toformula obtains $\psi_a \eqdef (x \eq g(6)\land f(x) \eq 6 \land y \eq 6)$.
For $\tgreps_{\ref{fig:tgreps-exb}}$, \toformula produces $\psi_b \eqdef (g(6) \eq x \land f(g(6)) \eq 6 \land y \eq 6)$. In some applications (like qelim considered in this paper) $\psi_b$ is preferred to $\psi_a$: simply removing the literals $g(6) \eq x$ and $y \eq 6$ from $\psi_b$ results in a formula equivalent to $\exists x, y \cdot
\varphi_{\ref{fig:tgreps-ex}}$ that does not contain variables.
Consider a third representative choice $\tgreps_{\ref{fig:tgreps-exc}}$, for node $\node(1)$, $\toexpr$ does not terminate: to produce a term for $\node(1)$, a term for $\node(3)$, the representative of its child, $\node(2)$, is required. Similarly to produce a term for $\node(3)$, a term for the representative of its child node $\node(5)$, $\node(1)$, is necessary. Thus, none of the terms can be extracted with $\tgreps_{\ref{fig:tgreps-exc}}$.

For extraction, representative functions $\tgreps$ are either provided explicitly or implicitly (as in~\cite{egg}), the latter by associating a cost to nodes and/or terms and letting the representative be a node with minimal cost. However, observe that not all costs guarantee that the chosen $\tgreps$ can be used (the computation does not terminate). For example, the ill-defined $\tgreps_{\ref{fig:tgreps-exc}}$ from above is a \repf that satisfies the cost function that assigns function applications cost 0 and variables and constants cost $1$.
A commonly used cost function is term AST size, which is sufficient to ensure termination of $\toexpr(n,\tgreps)$.

We are thus interested in characterizing representative functions motivated by two observations: not every cost function guarantees that $\nttrepr{n}$ terminates; and the kind of representative choices that are most suitable for \qelim ($\tgreps_{\ref{fig:tgreps-exb}}$) cannot be expressed over term AST size.

\begin{figure}[t]
  \centering
  \scalebox{0.9}{
    $\varphi_{\ref{fig:tgreps-ex}}(x,y) \eqdef y \eq f(x) \land x \eq g(y) \land f(x) \eq 6$}\\[-0.05in] 
  \rule{7cm}{.2mm}
  \\
  \hspace*{-2mm}
  \scalebox{0.9}{
    \centering
  \begin{subfigure}{0.33\textwidth}
  \begin{tikzpicture}[thick, node distance={9mm}, main/.style = {draw, circle,minimum height=5mm}]
    \node[main] (1) {$g$};
    \node[main] (2) [below of=1] {$y$};
    \node[main] (3) [right of=2] {$f$};
    \node[main] (4) [below of=3] {$x$};
    \node[main] (5) [right of=3] {$6$};
    \node [right of=1,xshift=-4.5mm,yshift=+2mm] {$^{(1)}$};
    \node [right of=2,xshift=-4.5mm,yshift=+2mm] {$^{(3)}$};
    \node [right of=3,xshift=-4.5mm,yshift=+2mm] {$^{(4)}$};
    \node [right of=4,xshift=-4.5mm,yshift=+2mm] {$^{(5)}$};
    \node [left of=2,xshift=4.5mm,yshift=+2mm] {$^{(2)}$};
\draw[->] (1) to (2);
\draw[->] (3) to (4);
    \draw[dashed,->,red] (1) to [out=230, in=180, looseness=1] (4);
    \draw[dashed,->,red] (2) -- (3);
    \draw[dashed,->,red] (5) -- (3);
\draw[dotted,->,blue] (1) to [out=330, in=100, looseness=1] (5);
    \draw[dotted,->,blue] (3) to [out=230, in=150, looseness=1] (4);
    \draw[dotted,->,white] (3) to [out=230, in=180, looseness=2.2] (1);
  \end{tikzpicture}
  \caption{$\tgreps_{\ref{fig:tgreps-exa}} \eqdef \{\node(4), \node(5) \}$\label{fig:tgreps-exa}}
  \end{subfigure}
\begin{subfigure}{0.33\textwidth}
  \begin{tikzpicture}[thick, node distance={9mm}, main/.style = {draw, circle,minimum height=5mm}]
    \node[main] (1) {$g$};
    \node[main] (2) [below of=1] {$y$};
    \node[main] (3) [right of=2] {$f$};
    \node[main] (4) [below of=3] {$x$};
    \node[main] (5) [right of=3] {$6$};
    \node [right of=1,xshift=-4.5mm,yshift=+2mm] {$^{(1)}$};
    \node [right of=2,xshift=-4.5mm,yshift=+2mm] {$^{(3)}$};
    \node [right of=3,xshift=-4.5mm,yshift=+2mm] {$^{(4)}$};
    \node [right of=4,xshift=-4.5mm,yshift=+2mm] {$^{(5)}$};
    \node [left of=2,xshift=4.5mm,yshift=+2mm] {$^{(2)}$};
\draw[->] (1) to (2);
\draw[->] (3) to (4);
    \draw[dashed,->,red] (1) to [out=230, in=180, looseness=1] (4);
    \draw[dashed,->,red] (2) -- (3);
    \draw[dashed,->,red] (5) -- (3);
\draw[dotted,->,white] (3) to [out=230, in=150, looseness=1] (4);
    \draw[dotted,->,blue] (1) to [out=330, in=100, looseness=1] (5);
    \draw[dotted,->,blue] (3) to [out=230, in=180, looseness=2.2] (1);
  \end{tikzpicture}
  \caption{$\tgreps_{\ref{fig:tgreps-exb}} \eqdef \{\node(4), \node(1) \}$\label{fig:tgreps-exb}}
  \end{subfigure}
\begin{subfigure}{0.33\textwidth}
    \begin{tikzpicture}[thick, node distance={9mm}, main/.style = {draw, circle,minimum height=5mm}]
      \node[main] (1) {$g$};
      \node[main] (2) [below of=1] {$y$};
      \node[main] (3) [right of=2] {$f$};
      \node[main] (4) [below of=3] {$x$};
      \node[main] (5) [right of=3] {$6$};
      \node [right of=1,xshift=-4.5mm,yshift=+2mm] {$^{(1)}$};
      \node [right of=2,xshift=-4.5mm,yshift=+2mm] {$^{(3)}$};
      \node [right of=3,xshift=-4.5mm,yshift=+2mm] {$^{(4)}$};
      \node [right of=4,xshift=-4.5mm,yshift=+2mm] {$^{(5)}$};
      \node [left of=2,xshift=4.5mm,yshift=+2mm] {$^{(2)}$};
\draw[->] (1) to (2);
\draw[->] (3) to (4);
      \draw[dashed,->,red] (1) to [out=230, in=180, looseness=1] (4);
      \draw[dashed,->,red] (2) -- (3);
      \draw[dashed,->,red] (5) -- (3);
\draw[dotted,->,white] (3) to [out=230, in=150, looseness=1] (4);
      \draw[dotted,->,blue] (1) to [out=330, in=100, looseness=1] (3);
      \draw[dotted,->,blue] (3) to [out=230, in=180, looseness=2.2] (1);
    \end{tikzpicture}
    \caption{$\tgreps_{\ref{fig:tgreps-exc}} \eqdef \{\node(3), \node(1) \}$\label{fig:tgreps-exc}}
    \end{subfigure}
  }
\caption{Egraphs of $\varphi_{\ref{fig:tgreps-ex}}$ with $G_\tgreps$.\label{fig:tgreps-ex}}
\end{figure}

\begin{definition}\label{def:validrepr}
Given an egraph $\egtuple$, a representative function $\tgreps : N \to N$ is \emph{admissible for $G$} if 
\begin{enumerate}
    \item[(a)] \tgreps assigns a unique representative per
class, 
\item[(b)] $\rho_\tgroot = \rho_\tgreps$, and
\item[(c)] the graph $G_\tgreps$ is acyclic, where $G_\tgreps = \langle \tgnodes, E_{\tgreps} \rangle$ and $ E_{\tgreps} \eqdef \{ (n, \tgreps(c)) \mid c \in \texttt{children}(n), n
\in \tgnodes\}$.
\end{enumerate}
\end{definition}

Dotted blue edges in the graphs of \cref{fig:tgreps-ex} show the corresponding $G_\tgreps$. Intuitively, for each node $n$, all reachable nodes in $G_{\tgreps}$ are the nodes whose \toexpr term is necessary to produce the $\nttrepr{n}$. Observe that $G_{\tgreps_{\ref{fig:tgreps-exc}}}$ has a cycle, thus, $\tgreps_{\ref{fig:tgreps-exc}}$ is not admissible.

\begin{reptheorem}{admissible}
\label{lemma:repr-termination}
  Given an egraph $G$ and a \repf \tgreps, the function $G.\toformula(\tgreps,\emptyset)$ terminates with result $\psi$ such that $\formula(G,\psi)$ iff $\tgreps$ is admissible
  for $G$.
\end{reptheorem}

To the best of our knowledge, \cref{lemma:repr-termination} is the first complete characterization of all terms of a node that can be obtained by extraction based on class representatives (via describing all admissible \tgreps, note that the number is finite). This result contradicts~\cite{egg}, where it is claimed to be possible to extract a term of a node for any cost function. The counterexample is $\tgreps_{\ref{fig:tgreps-exc}}$.
Importantly, this characterization allows us to explore representative functions outside those in the existing literature, which, as we show in the next section, is key for \qelim.

\section{Quantifier Reduction} \label{sec:qel}

Quantifier reduction is a relaxation of quantifier elimination: given two formulas $\varphi$ and $\psi$ with free variables \vars and \varsp, respectively, $\psi$ is a
\emph{quantifier reduction} of $\varphi$ if $\varsp
\subseteq \Vec{v}$ and $\varphi^{\exists} \equiv 
\psi^{\exists}$. 
If $\varsp$ is empty, then $\psi$ is a quantifier elimination of $\varphi^\exists$. 
Note that quantifier reduction is possible even when quantifier elimination is not (e.g., for EUF). We are interested in an efficient quantifier reduction algorithm (that can be used as pre-processing for \qelim), even if a complete \qelim is possible (e.g., for LIA). In this section, we present such an algorithm called \qel.

Intuitively, \qel is based on the well-known substitution rule: $(\exists x \cdot x \eq t \land \varphi) \equiv \varphi[x \mapsto t]$. A naive implementation of this rule, called \qelite in Z3, looks for syntactic definitions of the form $x \eq t$ for a variable $x$ and an $x$-free term $t$ and substitutes $x$ with $t$. While efficient, \qelite is limited because of: (a) dependence on syntactic equality in the formula (specifically, it misses implicit equalities due to transitivity and congruence); (b) sensitivity to the order in which variables are eliminated (eliminating one variable may affect available syntactic equalities for another); and (c) difficulty in dealing with circular~equalities~such~as~$x \eq f(x)$.

For example, consider the formula $\varphi_{\ref{fig:tgreps-ex}}(x, y)$ in \cref{fig:tgreps-ex}. Assume that $y$ is eliminated first using $y \eq f(x)$, resulting in 
$x \eq g(f(x)) \land f(x) \eq 6$. Now, $x$ cannot be eliminated since the only equality for $x$ is circular. Alternatively, assume that \qelite somehow noticed that by transitivity, $\varphi_{\ref{fig:tgreps-ex}}$ implies $y \eq 6$, and obtains $(\exists y \cdot \varphi_{\ref{fig:tgreps-ex}}) \eqdef x \eq g(6) \land f(x) \eq 6$. This time, $x \eq g(6)$ can be used to obtain $f(g(6)) \eq 6$ that is a \qelim of $\varphi_{\ref{fig:tgreps-ex}}^\exists$. Thus, both the elimination order and implicit equalities are crucial.

In \qel, we address the above issues by using an egraph data structure to concisely capture all implicit equalities and terms. Furthermore, egraphs allow eliminating multiple variables together, ensuring that a variable is eliminated if it is equivalent (explicitly or implicitly) to a ground term in the egraph.

\setlength{\intextsep}{5pt}

\begin{algorithm}[t]
  \caption{\qel\ -- Quantifier reduction using egraphs.}\label{alg:qe-lite}
  {\keywordfont Input:} A formula $\varphi$ with free variables $\vars$. \\ {\keywordfont Output:} A quantifier reduction of $\varphi$.
  \\[1mm]
  \scalebox{0.9}{
  \begin{minipage}{1\textwidth}
  \small
  $\mathit{\qel}(\varphi,\vars)$ \begin{algorithmic}[1]
    \State $G := \mathit{egraph}(\varphi)$ \label{l:qel1}
\State $\tgreps := G.\finddefs(\vars)$ \label{l:qel3}
    \State $\tgreps := G.\refinedefs(\tgreps,\vars)$ \label{l:qel3b}
    \State $\core := G.\findcore(\tgreps)$ \label{l:qel4}
    \State \Return $G.\toformula(\tgreps, G.\mathit{Nodes}() \setminus \core)$ \label{l:qel5}
  \end{algorithmic}
  \end{minipage}
  }
\end{algorithm}

Pseudocode for \qel is shown in \cref{alg:qe-lite}. Given an input
formula $\varphi$, \qel first builds its egraph $G$ (line~1). Then, it 
finds a representative function $\tgreps$ that maps variables to equivalent ground terms, as much as possible (line~2). Next, it further reduces the remaining free variables by refining $\tgreps$ to map each variable $x$ to an equivalent $x$-free (but not variable-free) term (line~\ref{l:qel3b}).
At this point, \qel is committed to the variables to eliminate. To produce the output, \findcore identifies the subset of the nodes of $G$, which we call \emph{core}, that must be considered in the output (line~\ref{l:qel4}). 
Finally, \toformula converts the core of $G$ to the resulting formula (line~\ref{l:qel5}). 
We show that the combination of these steps is even stronger than variable substitution.

To illustrate \qel, we apply it on $\varphi_{\ref{fig:egraph-example}}$ and its egraph $G$ from \cref{fig:egraph-example}. The function \finddefs returns  $\tgreps = \{\node(6),\node(8)\}$\footnote{Recall that we only show representatives of non-singleton classes.}. Node $\node(6)$ is the only node with a ground term in the equivalence class $\classof(\node(3))$. This corresponds to the definition $z \eq k + 1$.  Node $\node(8)$ is chosen arbitrarily since $\classof(\node(8))$ has no ground terms. There is no refinement possible, so $\refinedefs$ returns $\tgreps$. The core is $N  \setminus \{\node(3),\node(5),\node(9)\}$. Nodes $\node(3)$ and $\node(9)$ are omitted because they correspond to variables with definitions (under \tgreps), and $\node(5)$ is omitted because it is congruent to $\node(4)$ so only one of them is needed. Finally, $\toformula$ produces  $k + 1 \eq \mathit{read}(a,x) \land 3 > k + 1$. Variables $z$ and $y$ are eliminated.

In the rest of this section we present \qel in detail and \qel's key
properties.

\paragraph{Finding Ground Definitions.}\label{sec:choose-reps}
Ground variable definitions are found by selecting a \repf \tgreps that ensures that the maximum number of terms in the formula are rewritten into ground equivalent ones, which, in turn, means finding a ground definition for all variables that have one.

Computing a \repf \tgreps that is admissible and ensures finding ground definitions when they exist is not trivial.
Naive approaches for identifying ground terms, such as iterating arbitrarily over the classes and selecting a representative based on $\ntt(n)$ are not enough -- $\ntt(n)$ may not be in the output formula.
It is also not possible to make a choice based on $\nttrepr{n}$, since, in general, it cannot be yet computed (\tgreps is not known yet).

\begin{figure}[t]
  \centering
  \scalebox{0.9}{$\varphi_{\ref{ex:backtrack}}(x,y) \eqdef x \eq g(f(x)) \land
  y \eq h(f(y)) \land f(x) \eq f(y)$}\\
  \rule{7cm}{.2mm}
  \small
  \\[2mm]
  \scalebox{0.9}{
  \begin{subfigure}{0.36\textwidth}
    \centering
    \begin{tikzpicture}[thick, node distance={10mm}, main/.style = {draw, circle,minimum height=7mm}]
      \node[main] (1) {$g$};
      \node[main] (2) [below of=1] {$f$};
      \node[main] (3) [below of=2] {$x$};
      \node[main] (4) [right of=1] {$h$};
      \node[main] (5) [below of=4] {$f$};
      \node[main] (6) [below of=5] {$y$};

      \node [left of=1,xshift=+4mm,yshift=+3mm] {$^{(1)}$};
      \node [left of=2,xshift=+4mm,yshift=+3mm] {$^{(2)}$};
      \node [left of=3,xshift=+4mm,yshift=+3mm] {$^{(3)}$};
      \node [right of=4,xshift=-4mm,yshift=+3mm] {$^{(4)}$};
      \node [right of=5,xshift=-4mm,yshift=+3mm] {$^{(5)}$};
      \node [right of=6,xshift=-4mm,yshift=+3mm] {$^{(6)}$};

      \draw[->] (1) to (2);
      \draw[->] (2) to (3);
      \draw[->] (4) to (5);
      \draw[->] (5) to (6);

      \draw[dashed,->,red] (2) -- (5);
      \draw[dashed,->,red] (1) to [out=230, in=130, looseness=1] (3);
      \draw[dashed,->,red] (4) to [out=310, in=50, looseness=1] (6);

      \draw[dotted,->,blue] (2) to [out=60, in=300, looseness=1] (1);
      \draw[dotted,->,blue] (5) to [out=110, in=250, looseness=2] (4);
      \draw[dotted,->,blue] (1) -- (5);
      \draw[dotted,->,blue] (4) to [out=300, in=60, looseness=1] (5);
    \end{tikzpicture}
    \caption{$\tgreps_{\ref{fig:no-acyclic-choice}}\!=\!\{\node(1),\node(4),\node(5)\}$\label{fig:no-acyclic-choice}}
  \end{subfigure}
  \vline
    \begin{subfigure}{0.36\textwidth}
    \centering
    \begin{tikzpicture}[thick, node distance={10mm}, main/.style = {draw, circle,minimum height=7mm}]
      \node[main] (1) {$g$};
      \node[main] (2) [below of=1] {$f$};
      \node[main] (3) [below of=2] {$x$};
      \node[main] (4) [right of=1] {$h$};
      \node[main] (5) [below of=4] {$f$};
      \node[main] (6) [below of=5] {$y$};

      \node [left of=1,xshift=+4mm,yshift=+3mm] {$^{(1)}$};
      \node [left of=2,xshift=+4mm,yshift=+3mm] {$^{(2)}$};
      \node [left of=3,xshift=+4mm,yshift=+3mm] {$^{(3)}$};
      \node [right of=4,xshift=-4mm,yshift=+3mm] {$^{(4)}$};
      \node [right of=5,xshift=-4mm,yshift=+3mm] {$^{(5)}$};
      \node [right of=6,xshift=-4mm,yshift=+3mm] {$^{(6)}$};

      \draw[->] (1) to (2);
      \draw[->] (2) to (3);
      \draw[->] (4) to (5);
      \draw[->] (5) to (6);

      \draw[dashed,->,red] (2) -- (5);
      \draw[dashed,->,red] (1) to [out=230, in=130, looseness=1] (3);
      \draw[dashed,->,red] (4) to [out=310, in=50, looseness=1] (6);

      \draw[dotted,->,blue] (2) to [out=300, in=60, looseness=1] (3);
      \draw[dotted,->,blue] (5) to [out=300, in=65, looseness=1] (6);
      \draw[dotted,->,blue] (1) -- (5);
      \draw[dotted,->,blue] (4) to [out=300, in=60, looseness=1] (5);
    \end{tikzpicture}
    \caption{$\tgreps_{\ref{fig:only-leaves}}\!=\!\{\node(3),\node(6),\node(5)\}$\label{fig:only-leaves}}
  \end{subfigure}
\vline
  \begin{subfigure}{0.36\textwidth}
    \centering
    \begin{tikzpicture}[thick, node distance={10mm}, main/.style = {draw, circle,minimum height=7mm}]
      \node[main] (1) {$g$};
      \node[main] (2) [below of=1] {$f$};
      \node[main] (3) [below of=2] {$x$};
      \node[main] (4) [right of=1] {$h$};
      \node[main] (5) [below of=4] {$f$};
      \node[main] (6) [below of=5] {$y$};

      \node [left of=1,xshift=+4mm,yshift=+3mm] {$^{(1)}$};
      \node [left of=2,xshift=+4mm,yshift=+3mm] {$^{(2)}$};
      \node [left of=3,xshift=+4mm,yshift=+3mm] {$^{(3)}$};
      \node [right of=4,xshift=-4mm,yshift=+3mm] {$^{(4)}$};
      \node [right of=5,xshift=-4mm,yshift=+3mm] {$^{(5)}$};
      \node [right of=6,xshift=-4mm,yshift=+3mm] {$^{(6)}$};

      \draw[->] (1) to (2);
      \draw[->] (2) to (3);
      \draw[->] (4) to (5);
      \draw[->] (5) to (6);

      \draw[dashed,->,red] (2) -- (5);
      \draw[dashed,->,red] (1) to [out=230, in=130, looseness=1] (3);
      \draw[dashed,->,red] (4) to [out=310, in=50, looseness=1] (6);

      \draw[dotted,->,blue] (2) to [out=60, in=300, looseness=1] (1);
      \draw[dotted,->,blue] (5) to [out=300, in=65, looseness=1] (6);
      \draw[dotted,->,blue] (1) -- (5);
      \draw[dotted,->,blue] (4) to [out=300, in=60, looseness=1] (5);
    \end{tikzpicture}
    \caption{$\tgreps_{\ref{fig:qel-best-repr}}\!=\!\{\node(1),\node(6),\node(5)\}$\label{fig:qel-best-repr}}
  \end{subfigure}
  }
  \caption{Egraphs including \textcolor{blue}{$G_\tgreps$} of $\varphi_{\ref{ex:backtrack}}$.\label{ex:backtrack}}
\end{figure}

Admissibility raises an additional challenge since choosing a node that appears to be a definition (e.g., not a leaf) may cause cycles in $G_\tgreps$.
For example, consider $\varphi_{\ref{ex:backtrack}}$ of \cref{ex:backtrack}.  Assume that $\node(1)$ and $\node(4)$ are chosen as representatives of their equivalence classes. At this point, $G_\tgreps$ has two edges:  $\langle\node(5),\node(4)\rangle$ and $\langle\node(2), \node(1)\rangle$, shown by blue dotted lines in \cref{fig:no-acyclic-choice}. Next, if either $\node(2)$ or $\node(5)$ are chosen as representatives (the only choices in their class), then $G_\tgreps$ becomes cyclic (shown in blue in \cref{fig:no-acyclic-choice}).
Furthermore, backtracking on representative choices needs to be avoided if we are to find a \repf efficiently. 

\begin{algorithm}[t]
  \caption{Find definitions maximizing groundness.}\label{alg:ground-valid-repr}
  \small
  \centering
  \scalebox{0.9}{
  \begin{minipage}[t]{0.5\textwidth}
    $\egraph::\finddefs(\vars)$
    \begin{algorithmic}[1]
      \For{$n \in \tgnodes$}
      {$\tgreps(n) := \undefrepr$}
      \EndFor
      \State $\acq := \{\mathit{leaf}(n) \mid n \in N \land \mathit{ground}(n)\}$\label{ln:add-gr-leaf} \State $\tgreps := \processQ(\tgreps, \acq)$\label{ln:grndprocessq}
    \State $\acq := \{\mathit{leaf}(n) \mid n \in N\}$ \label{ln:add-leaf}\State $\tgreps := \processQ(\tgreps, \acq)$
    \State \Return $\tgreps$
    \algstore{find-rep}
  \end{algorithmic}
\end{minipage}\hspace{0.1in}
\begin{minipage}[t]{0.6\textwidth}
  $\egraph::\processQ(\tgreps, \acq)$
  \begin{algorithmic}[1]
    \algrestore{find-rep}
    \While{$\acq \neq \emptyset$}
    \State $n := \acq.\mathit{pop}()$
    \If{$\tgreps(n) \neq \undefrepr$}\label{ln:check-set-repr}
    {\kwf continue}
    \EndIf
    \For{$n' \in \classof(n)$ }
    $\tgreps(n') := n$\label{ln:chooserepr}
    \EndFor
\For{$n' \in \classof(n)$ }\label{ln:begin-check-parents}
    \For{$p \in \mathtt{parents}(n')$}
\If{$\forall c \in \children(p) \cdot \tgreps(c)\neq\undefrepr$}\label{ln:parentsel}
    \State $\acq.\mathit{push}(p)$
    \EndIf
    \EndFor
    \EndFor\label{ln:end-check-parents}
    \EndWhile
    \State \Return $\tgreps$
    \end{algorithmic}
  \end{minipage}
  }
\end{algorithm}

\cref{alg:ground-valid-repr} finds a \repf  $\tgreps$ while overcoming these challenges. To ensure that the computed \repf is admissible (without backtracking), \cref{alg:ground-valid-repr} selects representatives for each class using a ``bottom up'' approach. 
Namely, leaves cannot be part of cycles in $G_\tgreps$ because they have no outgoing edges. Thus, they can always be safely chosen as representatives. Similarly, a node whose children have already been assigned 
representatives in this way (leaves initially), 
will also never be part of a
cycle in $G_\tgreps$. Therefore, these nodes are also safe to be chosen as representatives. 

This intuition is implemented in $\finddefs$ by initializing $\tgreps$ to be undefined ($\undefrepr$) for all nodes, and maintaining a workset, \acq,
containing nodes that, if chosen for the remaining classes (under the current selection), maintain acyclicity of $G_\tgreps$.
The initialization of \acq includes leaves only. The specific choice of leaves ensures that ground definitions are preferred, and we return to it later. 
After initialization, the
function $\processQ$ extracts an element from \acq and sets it as the
representative of its class if the class has not been assigned yet
(lines~\ref{ln:check-set-repr} and~\ref{ln:chooserepr}).
Once a class representative has been chosen, on
lines~\ref{ln:begin-check-parents} to~\ref{ln:end-check-parents}, the parents of
all the nodes in the class such that all the children have been chosen (the
condition on line~\ref{ln:parentsel}) are added to \acq.
 
So far, we discussed how admissibility of $\tgreps$ is guaranteed. 
To also ensure that ground definitions are found whenever possible, 
we observe that a similar bottom up approach identifies terms that can be rewritten into ground ones.
This builds on the notion of constructively ground nodes, defined next.

A class $c$ is \emph{ground} if $c$ contains a \emph{constructively ground}, or \emph{c-ground} for short, node $n$, 
where a node $n$ is c-ground if either (a) $\ntt(n)$ is ground, or (b) $n$ is not a leaf and the class $\classof(n[i])$ of every child $n[i]$ is ground.
Note that nodes labeled by variables are never \cground.

In the example in \cref{fig:egraph-example}, $\classof(\node(7))$ and $\classof(\node(8))$ are not ground, because all their nodes represent variables; $\classof(\node(6))$ is ground because $\node(6)$ is c-ground. Nodes $\node(4)$ and $\node(5)$ are not \cground because the class of $\node(8)$ (a child of both nodes) is not ground. Interestingly, $\node(1)$ is \cground,
because $\classof(\node(3)) = \classof(\node(6))$ is ground, even though its term $3 > z$ is~not~ground.  

Ground classes and c-ground nodes are of interest because whenever $\varphi \models \ntt(n) \eq t$ for some node $n$ and ground term $t$, 
then $\classof(n)$ is ground, i.e., 
it contains a c-ground node, where c-ground nodes can be found recursively starting from ground leaves. 
Furthermore, the recursive definition ensures that when the aforementioned c-ground nodes 
are selected as representatives, the corresponding terms w.r.t. \tgreps are ground.

As a result, to maximize the ground definitions found, we are interested in finding an admissible \repf $\tgreps$ that is \emph{maximally ground}, which means that
for every node $n \in N$, if $\classof(n)$ is ground, then $\tgreps(n)$ is \cground. 
That means that \cground nodes are always chosen if they exist.

\begin{reptheorem}{thm:gr-defs-new}Let $G=\egraph(\varphi)$ be an egraph and $\tgreps$ an admissible \repf that is maximally ground. 
  For all $n \in N$, if $\varphi \models \ntt(n) \eq t$ for some ground term $t$, then $\tgreps(n)$ is c-ground and $\nttrepr{\tgreps(n)}$ is ground.
\end{reptheorem}

We note that not every choice of c-ground nodes as representatives results in an admissible \repf. 
For example, consider the formula $\varphi_{\ref{fig:tgreps-ex}}$ of \cref{fig:tgreps-ex} and its egraph. 
All nodes except for $\node(5)$ and $\node(2)$ are c-ground. However, a \tgreps with $\node(3)$ and $\node(1)$ as representatives is not admissible.
Intuitively, this is because the ``witness'' for c-groundness of $\node(1)$ in $\classof(\node(2))$ is $\node(4)$ and not $\node(3)$. Therefore, it is important to incorporate the selection of c-ground representatives 
into the bottom up procedure that ensures admissibility of $\tgreps$.

To promote \cground nodes over non \cground in the construction of an admissible \repf,
$\finddefs$ chooses representatives in two steps. First,
only the ground leaves are processed (line~\ref{ln:add-gr-leaf}). This ensures that \cground representatives are chosen
while guaranteeing the absence of cycles. Then, the remaining leaves are added
to \acq (line~\ref{ln:add-leaf}). This triggers representative selection
of the remaining classes (those that are not ground).

We illustrate $\finddefs$ with two examples. For
$\varphi_{\ref{fig:tgreps-ex}}$ of \cref{fig:tgreps-ex}, there is only one leaf
that is ground, $\node(4)$, which is added to \acq on line~\ref{ln:add-gr-leaf},
and \acq is processed. $\node(4)$ is chosen as representative and, as a
consequence, its parent $\node(1)$ is added to \acq. $\node(1)$ is chosen as
representative so $\node(3)$, even though added to the queue later, is not
chosen as representative, obtaining $\tgreps_{\ref{fig:tgreps-ex}b} = \{\node(4),\node(1)\}$.
For $\varphi_{\ref{ex:backtrack}}$ of \cref{ex:backtrack}, 
no nodes are added to \acq on line~\ref{ln:add-gr-leaf}. $\node(3)$ and $\node(6)$
are added on line~\ref{ln:add-leaf}. In $\processQ$, both are chosen as
representatives obtaining, $\tgreps_{\ref{fig:only-leaves}}$.

\cref{alg:ground-valid-repr} guarantees that $\tgreps$ is maximally ground.
Together with \Cref{thm:gr-defs-new}, this implies that all terms that can be rewritten into ground equivalent ones will be rewritten, which, in turn, means that for each variable that has a ground definition, its representative is one such definition.

\paragraph{Finding Additional (Non-ground) Definitions.}
At this point, \qel found ground definitions while avoiding cycles in $G_\tgreps$. However, this does not mean that as many variables as possible are eliminated. A variable can also be eliminated if it can be expressed as a function of other variables. This is not achieved by \finddefs. For example, in $\tgreps_{\ref{fig:only-leaves}}$ both variables are representatives, hence none is eliminated, even though,  
since $x \eq g(f(y))$, $x$ could be eliminated in $\varphi_5$ by rewriting $x$ as a function of $y$.
\cref{alg:var-elim} shows function $\refinedefs$ that refines maximally ground $\tgreps{s}$ to further find such definitions while keeping admissibility and ground maximality. This is done by greedily attempting to change class representatives if they are labeled with a variable.
$\refinedefs$ iterates over the nodes in the class checking if there is a different node that is not a variable and that does not create a cycle in $G_\tgreps$ (line~\ref{ln:makescycle}).
The resulting $\tgreps$ remains maximally ground because representatives of ground classes are not changed.

\begin{algorithm}[t]
  \caption{Refining \tgreps and building core.}\label{alg:var-elim}\label{alg:redundant}
  \small
  \centering
  \scalebox{0.9}{
  \begin{minipage}[t]{0.44\textwidth}
    $\egraph::\refinedefs(\tgreps, \vars)$
  \begin{algorithmic}[1]
    \For{$n \in N$}
    \If{$n = \tgreps(n)$ {\kwf and} $L(n) \in \vars$}
    \State $r := n$
    \For{$n' \in \classof(n) \setminus \{n\}$}
    \If{$L(n') \not\in \vars$}
    \If{{\kwf not} $\texttt{cycle}(n',\tgreps)$}
    \label{ln:makescycle}
    \State $\mathit{r} := n'$; 
    \State {\kwf break}
    \EndIf
    \EndIf
    \EndFor
    \For{$n' \in \classof(n)$}
    \State $\tgreps[n'] := r$
    \EndFor
\EndIf
    \EndFor
    \State \Return $\tgreps$
    \end{algorithmic}
\end{minipage}\hspace{0.1in}
\begin{minipage}[t]{.56\textwidth}
  $\egraph::\findcore(\tgreps,\vars)$
  \begin{algorithmic}[1]
    \State $\core := \emptyset$
    \For{$n \in \tgnodes$ s.t. $n = \tgreps(n)$} \State $\core := \core \cup \{n\}$
\For{$n' \in (\classof(n) \setminus n)$}
    \If{$\tglabel(n') \in \vars$}
    {\kwf continue}
\ElsIf{$\exists m \in \core \cdot m \text{ congruent with }n'$ \Statex \hspace*{5mm}} 
    \State {\kwf continue}
\EndIf
\State $\core := \core \cup \{n'\}$
    \EndFor
    \EndFor
    \State \Return $\core$
  \end{algorithmic}
\end{minipage}
  }
\end{algorithm}

For example, let us refine $\tgreps_{\ref{fig:only-leaves}} = \{\node(3),\node(6),\node(5)\}$ obtained for $\varphi_{\ref{ex:backtrack}}$.
Assume that $x$ is processed first. For $\classof(\node(x))$, changing the representative to $\node(1)$
does not introduce a cycle (see \cref{fig:qel-best-repr}), so $\node(1)$ is selected.
Next, for $\classof(\node(y))$, 
choosing $\node(4)$ causes $G_\tgreps$ to be cyclic since  $\node(1)$ was already chosen (\cref{fig:no-acyclic-choice}), so the representative of $\classof(\node(y))$ is not changed.
The final refinement is $\tgreps_{\ref{fig:qel-best-repr}} = \{\node(1),\node(6),\node(5)\}$.

At this point, \qel found a \repf $\tgreps$ with as many ground definitions as possible and attempted to refine $\tgreps$ to have fewer variables as representatives. Next, \qel finds a core of the nodes of the egraph, based on $\tgreps$, that will govern the translation of the egraph to a formula. While \tgreps determines the semantic rewrites of terms that enable variable elimination, it is the use of the core in the translation that actually eliminates them.

\paragraph{Variable Elimination Based on a Core.}\label{sec:redundant}
A \emph{core} of an egraph $\egtuple$ and a \repf $\tgreps$, is a subset of the nodes $N_c \subseteq N$ such that $\psi_\textit{c} = G.\toformula(\tgreps,N\setminus N_c)$ satisfies $\formula(G,\psi_\textit{c})$. 

\cref{alg:redundant} shows pseudocode for $\findcore$ that computes a core of an egraph for a given representative function. 
The idea is that non-representative nodes that are labeled by variables, as well as nodes congruent to nodes that are already in the core, need not be included in the core.
The former are not needed since we are only interested in preserving the existential closure of the output, while the latter are not needed since congruent nodes introduce the same syntactic terms in the output.
For example, for $\varphi_1$ and $\tgreps_1$, $\findcore$ returns $\core_1 = N_1 \setminus \{\node(3), \node(5), \node(9)\}$. Nodes $\node(3)$ and $\node(9)$ are excluded because they are labeled with variables; and node $\node(5)$ because it is congruent with $\node(4)$.

Finally, \qel produces a quantifier reduction by applying $\toformula$ with the computed $\tgreps$ and $\core$. Variables that are not in the core (they are not representatives) are eliminated --
this includes variables that have a ground definition. However,
\qel may eliminate a variable even if it is a representative (and thus it is in the core).
As an example, consider $\psi(x,y) \eqdef f(x) \eq f(y) \land x \eq y$, whose egraph $G$ contains 2 classes with 2 nodes each. The core $N_c$ relative to any admissible $\tgreps$ contains only one representative per class: in the $\classof(\node(x))$ because both nodes are labeled with variables, and in the $\classof(\node(f(x)))$ because nodes are congruent. In this case, $\toformula(\tgreps,N_c)$ results in $\top$ (since singleton classes in the core produce no literals in the output formula),
a quantifier elimination of $\psi$.
More generally, the variables are eliminated because none of them is reachable in $G_\tgreps$ from a non-singleton class in the core (only such classes contribute literals to the output).

We conclude the presentation of \qel by showing its output for our examples. For $\varphi_1$, \qel obtains \mbox{$(k + 1 \eq \mathit{read}(a,x) \land 3 > k + 1)$}, a quantifier reduction, using $\tgreps_1 = \{\node(3), \node(8))\} $ and $\core_1 = N_1 \setminus \{\node(3), \node(5), \node(9)\}$. For $\varphi_{\ref{fig:tgreps-ex}}$, \qel obtains $(6 \eq f(g(6)))$, a quantifier elimination, using $\tgreps_{\ref{fig:tgreps-ex}b} =
\{\node(4),\node(1)\}$, and  $\core_{4b} = N_{4} \setminus \{\node(3), \node(2)\}$. Finally, for $\varphi_{\ref{ex:backtrack}}$, \qel obtains $(y \eq h(f(y)) \land f(g(f(y))) \eq f(y))$, a quantifier reduction, using $\tgreps_{\ref{fig:qel-best-repr}}
= \{\node(1),\node(6),\node(5)\}$ and $\core_{5c} = N_5\setminus \{\node(3)\}$.

\paragraph{Guarantees of \qel.}
Correctness of \qel is straightforward. We conclude this section by providing two conditions that ensure that a variable is eliminated by \qel. The first condition guarantees that a variable is eliminated whenever a ground definition for it exists (regardless of the specific \repf and core computed by \qel). This makes  \qel \emph{complete relative to quantifier elimination based on ground definitions}.  
Relative completeness is an important property since it means that 
\qel is unaffected by variable orderings and syntactic rewrites, unlike \qelite.
The second condition, illustrated by $\psi$ above, depends on the specific \repf and core computed by \qel. 
\begin{reptheorem}{qelsummary}
\label{thm:varelim-summary}
  Let $\varphi$ be a QF conjunction of literals with free variables $\vars$, and let $v \in \vars$. 
    Let $G = \egraph(\varphi)$, $n_v$ the node in $G$ such that $L(n_v) = v$ and $\tgreps$ and \core computed by \qel. We denote by $\mathit{NS} = \{n \in \core \mid (\classof(n) \cap \core) \neq \{n\} \}$ the set of nodes from classes with two or more nodes in \core.
  If one of the following conditions hold, then $v$ does not appear in $\mathit{\qel}(\varphi,\vars)$:
  \begin{enumerate}
  \item[(1)] there exists a ground term $t$ s.t. $\varphi \models v \eq t$, or
  \item[(2)] $n_v$ is not reachable from any node in $\mathit{NS}$ in $G_\tgreps$.
  \end{enumerate}
\end{reptheorem}
As a corollary, if every variable meets one of the two conditions, then \qel finds a quantifier elimination.

This concludes the presentation of our quantifier reduction algorithm. Next, we show how \qel can be used to under-approximate quantifier elimination, which allows working with formulas for which \qel does not result in a \qelim.

\section{Model Based Projection Using QEL}\label{sec:mbp}
\begin{figure}[t]
\scalebox{0.9}{
  \begin{minipage}{0.6\textwidth}
    \centering
\begin{tabular}{l}
$\inferrule*[lab=\textsc{ElimWrRd1}, right=$M\models i \eq j$]{\varphi[\arread(\arwrite(t,i, v),j)]}{\varphi[v] \land i \eq j}$\\ \\
$\inferrule*[lab=\textsc{ElimWrRd2}, right=$M\models i \deq j$]{\varphi[\arread(\arwrite(t,i, v),j)]}{\varphi[\arread(t, j)]\land i \deq j}$
\end{tabular}
\caption{Two MBP rules from~\cite{DBLP:conf/fmcad/KomuravelliBGM15}. The notation $\varphi[t]$ means that $\varphi$ contains term $t$. The rules rewrite all occurrences of $\arread(\arwrite(t,i, v),j)$ with $v$ and $\arread(t, j)$, respectively. \label{fig:arr-rd-wr-mbp}}
\end{minipage}\hspace{0.4cm}
\begin{minipage}{0.45\textwidth}
  \small
  $\textit{ElimWrRd}$
  \begin{algorithmic}[1]
    \Function{\checkr}{$t$}
    \State \Return $t = \arread(\arwrite(s, i, v), j)$
    \EndFunction
  \Function{\applyr}{$t, M, G$}
  \If {$M \models i \eq j$}
  \State $G.\textit{assert}(i \eq j)$
  \State $G.\textit{assert}(t \eq v)$
  \Else
  \State $G.\textit{assert}(i \deq j)$
  \State $G.\textit{assert}(t \eq \arread(s, j))$
  \EndIf
  \EndFunction
  \end{algorithmic}
\caption{Adaptation of rules in \cref{fig:arr-rd-wr-mbp} using QEL
  API.\label{fig:arr-rd-wr-egraph}}
\end{minipage}
}
\end{figure}

Applications like model checking and quantified satisfiability require
efficient computation of under-approximations of quantifier
elimination.
They make use of model-based projection~(MBP) algorithms
to project variables that cannot be eliminated cheaply. 
Our \qel algorithm is efficient and relatively complete, but it does not guarantee to eliminate all variables. 
In this section, we use a model and theory-specific
projection rules to implement an MBP algorithm on~top~of~\qel. 

We focus on two important theories: Arrays and Algebraic
DataTypes~(ADT). They are widely used to encode program verification
tasks. Prior works separately develop MBP algorithms for
Arrays~\cite{DBLP:conf/fmcad/KomuravelliBGM15} and
ADTs~\cite{DBLP:conf/lpar/BjornerJ15}. Both MBPs were presented as a
set of syntactic rewrite rules applied until fixed point.

Combining the MBP algorithms for Arrays and ADTs is non-trivial 
because applying projection rules for one theory may produce terms of
the other theory. Therefore, separately achieving saturation in either
theory is not sufficient to reach saturation in the combined
setting. The MBP for the combined setting has to call both
MBPs, check whether either one of them produced terms that
can be processed by the other, and, if so, call the other
algorithm. This is similar to theory combination in SMT solving where
the core SMT solver has to keep track of different theory solvers and
exchange terms between them.

Our main insight is that egraphs can be used as a glue to combine MBP
algorithms for different theories, just like egraphs are used in SMT
solvers to combine satisfiability checking for different
theories. Implementing MBP using egraphs allows us to use the insights
from \qel to combine MBP with on-the-fly quantifier reduction to
produce less under-approximate formulas than what we get by syntactic
application of MBP rules. 

To implement MBP using egraphs, we implement all rewrite rules for MBP
in Arrays~\cite{DBLP:conf/fmcad/KomuravelliBGM15} and
ADTs~\cite{DBLP:conf/lpar/BjornerJ15} on top of egraphs. In the
interest of space, we explain the implementation of just a couple of
the MBP rules for Arrays\footnote{Implementation of all other rules is
similar. \condappendix{See \cref{sec:mbprules} for details.}}.

\cref{fig:arr-rd-wr-mbp} shows two Array MBP rules
from~\cite{DBLP:conf/fmcad/KomuravelliBGM15}: \textsc{ElimWrRd1} and
\textsc{ElimWrRd2}. Here, $\varphi$ is a formula with arrays and $M$
is a model for $\varphi$. Both rules rewrite terms which match the
pattern $\mathit{read}(\mathit{write}(t, i, v), j)$, where $t$, $i$,
$j$, $k$ are all terms and $t$ contains a variable to be
projected. \textsc{ElimWrRd1} is applicable when $M \models i \eq
j$. It rewrites the term $\mathit{read}(\mathit{write}(t, i, v), j)$
to $v$. \textsc{ElimWrRd2} is applicable when $M\not\models i \eq j$
and rewrites $\mathit{read}(\mathit{write}(t, i, v), j)$ to
$\mathit{read}(t, j)$.

\cref{fig:arr-rd-wr-egraph} shows the egraph implementation of \textsc{ElimWrRd1} and \textsc{ElimWrRd2}. The $\checkr(t)$ method
checks if $t$ syntactically matches $\mathit{read}(\mathit{write}(s,
i, v), j)$, where $s$ contains a variable to be
projected. The $\applyr(t)$ method assumes that $t$ is
$\mathit{read}(\mathit{write}(s, i, v), j)$. It first checks if
$M\models i \eq j$, and, if so, it adds $i\eq j$ and $t\eq v$
to the egraph $G$. Otherwise, if $M \not\models i \eq j$, $\applyr(t)$ adds a
disequality $i \deq j$ and an equality $t \eq \mathit{read}(s,
v)$ to $G$. That is, the egraph implementation of the rules only adds
(and does not remove) literals that capture the side condition and the
conclusion of the rule.

Our algorithm for MBP based on egraphs, \egraphmbp, is shown in \cref{alg:mbp}. It
initializes an egraph with the input
formula~(line~\ref{ln:egraphinit}), applies MBP rules until
saturation~(line~\ref{ln:applyrules}), and then uses the steps of
\qel~(lines~\ref{ln:mbprepr}--\ref{ln:mbpabs}) to generate the
projected formula.

Applying rules is as straightforward as iterating over all terms $t$
in the egraph, and for each rule $r$ such that $r.\checkr(t)$ is true,
calling $r.\applyr(t, M,
G)$~(lines~\ref{ln:ubegin}--\ref{ln:uend}). As opposed to the standard
approach based on formula rewriting, here the terms are \emph{not}
rewritten -- both remain. Therefore, it is possible to get into an
infinite loop by re-applying the same rules on the same terms over and
over again. 
To avoid this, \egraphmbp marks terms as
\emph{seen}~(line~\ref{ln:markuseen}) and avoids them in the next
iteration~(line~\ref{ln:avoiduseen}). Some rules in MBP are applied to
pairs of terms. For example, \textsc{Ackermann}\condappendix{~(defined in
\cref{sec:mbprules})} rewrites pairs of $\mathit{read}$ terms over the same
variable. This is different from usual applications where rewrite
rules are applied to individual expressions. Yet, it is easy to adapt
such pairwise rewrite rules to egraphs by iterating over pairs of
terms~(lines~\ref{ln:avoidpseen}--\ref{ln:markpseen}).

\egraphmbp does not apply MBP rules to terms that contain variables
but are already \cground~(line~\ref{ln:avoidcground}), which is sound
because such terms are replaced by ground terms in the
output~(\cref{thm:varelim-summary}). This prevents unnecessary
application of MBP rules thus allowing
\egraphmbp to compute MBPs that are closer to a quantifier elimination
(less model-specific).

Just like each application of a rewrite rule introduces a new term to
a formula, each call to the $\applyr$ method of a rule adds new terms
to the egraph. Therefore, each call to
$\mathit{ApplyRules}$~(line~\ref{ln:applyrules}) makes the egraph
bigger. 
However, provided that the original MBP combination is terminating,
the iterative application of $\mathit{ApplyRules}$ terminates as well
(due to marking).

Some MBP rules introduce new variables to the formula. \egraphmbp
computes $\tgreps$ based on both original and newly introduced
variables~(line~\ref{ln:mbprepr}). This allows \egraphmbp to eliminate
all variables, including non-Array, non-ADT variables, that are
equivalent to ground terms~(\cref{thm:varelim-summary}).

As mentioned earlier, \egraphmbp never removes terms while rewrite rules are
saturating. Therefore, after saturation, the egraph still contains all
original terms and variables. From soundness of the MBP rules, it
follows that after each invocation of $\applyr$, \egraphmbp creates an
under-approximation of $\varphi^\exists$ based on the model $M$. From
completeness of MBP rules, it follows that, after saturation, all
terms containing Array or ADT variables can be removed from the egraph
without affecting equivalence of the saturated egraph. Hence, when
calling $\toformula$, \egraphmbp removes all terms containing Array or
ADT variables~(line~\ref{ln:mbpabs}). This includes, in particular,
all the terms on which rewrite rules were applied, but potentially
more.

\begin{algorithm}[t]
  {\keywordfont Input:} A QF formula $\varphi$ with free variables
  $\Vec{v}$ all of sort $\textit{Array}(I, V)$ or $\textit{ADT}$,
a model $M \models \varphi^{\exists}$, and sets of rules $\textit{ArrayRules}$ and $\textit{ADTRules}$\\
  {\keywordfont Output:} A cube $\psi$
  s.t. $\psi^\exists \limp \varphi^\exists$, $M \models \psi^\exists$, and
  $\textit{vars}(\psi)$ are not Arrays or ADTs\\[2mm]
  \small
  \scalebox{0.9}{
  \begin{minipage}[t]{0.54\textwidth}
      $\egraphmbp(\varphi, \vars, M)$
          \begin{algorithmic}[1]
\State $G := \egraph(\varphi)$\label{ln:egraphinit}
      \State $p_1, p_2 := \top, \top;$ $\notcore, \notcore_p := \emptyset, \emptyset$
      \While{$p_1\lor p_2$}
      \State $p_1 := \textit{ApplyRules}(G, M, \textit{ArrayRules}, \notcore, \notcore_p)$\label{ln:applyrules}
      \State $p_2 := \textit{ApplyRules}(G, M, \textit{ADTRules}, \notcore, \notcore_p)$
      \EndWhile
        \State $\vars' := G.\textit{Vars}()$
        \State $\tgreps := G.\finddefs(\vars')$\label{ln:mbprepr}
        \State $\tgreps := G.\refinedefs(\tgreps, \vars')$        
        \State $\core := G.\findcore(\tgreps, \vars')$
        \State $\vars_e := \{v \in \vars' \mid \textit{is\_arr}(v)\lor \textit{is\_adt}(v)\}$
        \State $\core_e := \{n \in \core \mid \textit{gr}(\textit{term}(n), \vars_e)\}$
\State \Return $G.\toformula(\tgreps, G.\textit{Nodes}() \setminus\core_e)$\label{ln:mbpabs}
\algstore{end-applybrules}
    \end{algorithmic}
    \end{minipage}\hspace{0.2in}
    \begin{minipage}[t]{0.44\textwidth}
    $\mathit{ApplyRules}(G, M, R, \notcore, \notcore_p)$
  \begin{algorithmic}[1]
            \algrestore{end-applybrules}
\State $\textit{progress} := \bot$
    \State $N := G.\textit{Nodes}()$\label{ln:ubegin}
    \State $U := \{n \mid n \in N\setminus \notcore\}$\label{ln:avoiduseen}
    \State $T := \{\textit{term}(n) \mid n \in U \land{}$\label{ln:avoidcground}
    \Statex \hfill $(\mathit{is\_eq}(\mathit{term}(n)) \lor \neg \textit{\cground}(n))\}$
    \State $R_p := \{ r \in R\mid r.\mathit{is\_for\_pairs}()\}$
    \State $R_u := R\setminus R_p$
      \ForEach {$t \in T, r \in R_u$}\label{ln:forallterms}
        \If{$r.\checkr(t)$}\label{ln:checkrule}
        \State $r.\applyr(t, M, G)$\label{ln:applyr}
        \State $\textit{progress} := \top$
        \EndIf
        \EndFor\label{ln:uend}
        \State $\notcore := \notcore \cup N$\label{ln:markuseen}
        \State $N_p := \{ \langle n_1, n_2\rangle \mid n_1, n_2 \in N\}$\label{ln:pbegin}
        \State $T_p := \{\textit{term}(n_p)  \mid n_p \in N_p \setminus \notcore_p\}$\label{ln:avoidpseen}
        \ForEach {$t_p \in T_p, r \in R_p$}\label{ln:forallbterms}
        \If{$r.\checkr(p)$}\label{ln:collectbterms}
        \State $r.\applyr(p, M, G)$\label{ln:applybterms}
        \State $\textit{progress} := \top$\label{ln:btprogress}
        \EndIf
        \EndFor\label{ln:pend}
        \State $\notcore_p := \notcore_p \cup N_p$\label{ln:markpseen}
        \State \Return $\textit{progress}$
\end{algorithmic}
    \end{minipage}
    }
  \caption{\egraphmbp: an MBP using \qel. Here $\textit{gr}(t, \vars)$ checks
    whether term $t$ contains any variables in $\vars$ and $\mathit{is\_eq}(t)$ checks if $t$ is an equality literal.\label{alg:mbp}}
\end{algorithm}

We demonstrate our MBP algorithm on an example with nested ADTs and
Arrays. Let $\pairsort \eqdef \langle \intarr, \intsort\rangle$ be the
datatype of a pair of an integer array and an integer, and let $\pair:
\intarr\times \intsort \rightarrow \pairsort$ be its sole constructor
with destructors $\fst: \pairsort \rightarrow \intarr$ and $\snd:
\pairsort \rightarrow \intsort$. In the following, let $i$, $l$, $j$
be integers, $a$ an integer array, $p$, $p'$ pairs, and $\ap_1$,
$\ap_2$ arrays of pairs~($A_{\intsort\times\pairsort}$).~Consider~the~formula:
$$\varphi_{\mathit{mbp}}(p,a) \;\eqdef\; \mathit{read}(a, i) \eq i
\land p \eq \pair(a, l) \land \ap_2 \eq \mathit{write}(\ap_1, j,
p)\land p \deq p'$$ 
where $p$ and $a$ are free variables that we want
to project and all of $i, j, l, \ap_1, \ap_2, p'$ are constants that
we want to keep. MBP is guided by a model $M_{\mathit{mbp}}\models
\varphi_{\mathit{mbp}}$. To eliminate $p$
and $a$, \egraphmbp constructs the egraph of $\varphi_{\mathit{mbp}}$
and applies the MBP rules. In particular, it uses Array
MBP rules to rewrite the $\mathit{write}(\ap_1, j, p)$ term by adding
the equality $\mathit{read}(\ap_2, j) \eq p$ and merging it with the
equivalence class of $\ap_2 \eq \mathit{write}(\ap_1, j, p)$. It then
applies ADT MBP rules to deconstruct the equality $p \eq
\mathit{pair}(a, l)$ by creating two equalities $\fst(p) \eq a$ and
$\snd(p) \eq l$. Finally, the call to $\toformula$ produces\vspace{-0.1in}\begin{multline*}
  \arread(\fst(\arread(\ap_1, j)), i) \eq i \land \snd(\arread(\ap_1, j)) \eq l \land{}\\
  \arread(\ap_2, j) \eq \pair(\fst(\arread(\ap_1, j)), l) \land{}\\
  \ap_2 \eq \arwrite(\ap_1, j, \arread(\ap_2, j))\land \arread(\ap_2, j) \deq p'
\end{multline*}
The output is easy to understand by tracing it back to the
input. For example, the first literal is a rewrite of the literal
$\mathit{read}(a, i) \eq i$ where $a$ is represented with $\fst(p)$ and $p$ is
represented with $\mathit{read}(\ap_1, j)$. While the interaction of these
rules might seem straightforward in this example, the MBP
implementation in \zthree fails to project $a$ in this example because of the
multilevel nesting.

Notably, in this example, the \cground computation during projection
allows \egraphmbp not splitting on the disequality $p\deq p'$ based
on the model.  While ADT MBP rules eliminate disequalities by
using the model to split them, \egraphmbp benefits from the fact that,
after the application of Array MBP rules, the class of $p$ becomes
ground, making $p\deq p'$ \cground.  Thus, the \cground computation
allows \egraphmbp to produce a formula that is less approximate than
those produced by syntactic application of MBP rules. In fact, in this
example, a quantifier elimination is obtained~(the model
$M_{\mathit{mbp}}$ was not used).

In the next section, we show that our improvements to MBP translate to
significant improvements in a CHC-solving procedure that relies on MBP.

\section{Evaluation}\label{sec:evaluation}

\newcommand{\yicesQS}{\textsc{YicesQS}\xspace}

We implement \qel~(\cref{alg:qe-lite}) and \egraphmbp~(\cref{alg:mbp}) inside \zthree~\cite{z3} (version 4.12.0), a state-of-the-art SMT solver. Our implementation (referred to as \ztg), is publicly available on GitHub\footnote{Available at \url{https://github.com/igcontreras/z3/tree/qel-cav23}.}. \ztg replaces \qelite with
\qel, and the existing MBP with \egraphmbp.

We evaluate \ztg using two solving tasks.
Our first evaluation is on the QSAT algorithm~\cite{DBLP:conf/lpar/BjornerJ15} for
checking satisfiability of formulas with alternating
quantifiers. In QSAT, \zthree uses both \qelite and MBP to
under-approximate quantified formulas.  
We compare three QSAT implementations: the existing version in \zthree with the default \qelite and MBP; the existing version in \zthree in which \qelite and MBP are replaced by our egraph-based algorithms, \ztg; and the QSAT implementation in \yicesQS\footnote{Available at \url{https://github.com/disteph/yicesQS}.}, based on the \textsc{Yices}~\cite{DBLP:conf/cav/Dutertre14} SMT solver. 
During the evaluation, we found a bug in QSAT implementation of \zthree and fixed it\footnote{Available at \url{https://github.com/igcontreras/z3/commit/133c9e438ce}.}. The fix resulted in \zthree solving over $40$ sat instances and over $120$ unsat instances more than before. In the following, we use the fixed version of \zthree.

We use benchmarks in the theory of (quantified) LIA and LRA
from SMT-LIB~\cite{BarFT-SMTLIB,BarST-SMT-10}, with alternating
quantifiers.  LIA and LRA are the only tracks in which \zthree uses the QSAT tactic by default. To make our experiments more comprehensive, we also consider two modified variants of the LIA and LRA benchmarks, where we add some non-recursive ADT variables to the benchmarks. Specifically, we wrap all existentially quantified arithmetic variables using a record type ADT and unwrap them whenever they get used\footnote{The modified benchmarks are available at \url{https://github.com/igcontreras/LIA-ADT} and \url{https://github.com/igcontreras/LRA-ADT}.}. Since these benchmarks are similar to the original, we force \zthree to use the QSAT tactic on them with a \texttt{tactic.default\_tactic=qsat} command line option. 

\Cref{tab:qsat1} summarizes the results for the SMT-LIB benchmarks. In LIA, both \ztg and \zvanilla solve all benchmarks in under a minute, while \yicesQS is unable to solve many instances.
In LRA, \yicesQS solves all instances with very good performance. 
\zvanilla is able to solve only some benchmarks, and our \ztg performs similarly to \zvanilla. 
We found that in the LRA benchmarks, the new algorithms in \ztg 
are not being used since there are not many equalities in the formula, and no equalities are inferred during the run of QSAT. Thus, any differences between \zthree and \ztg are due to inherent randomness of the solving process.

\Cref{tab:qsat2} summarizes the results for the categories of mixed ADT and arithmetic. \yicesQS is not able to compete because it does not support ADTs. As expected, \ztg solves many more instances than \zvanilla.

\definecolor{Gray}{gray}{0.85}
\newcolumntype{a}{>{\columncolor{Gray}}r}

\begin{table}[t]
\begin{minipage}[t]{0.5\textwidth}
\centering
    \scalebox{.85}{
    \setlength{\tabcolsep}{1pt}
  \begin{tabular}{l r @{\hspace{0.1in}} a a @{\hspace{0.1in}} r r @{\hspace{0.1in}} r r}
        \toprule
       \rowcolor{white}
       \multirow{2}{*}{Cat.} & \multirow{2}{*}{Count} & \multicolumn{2}{c}{\ztg} & \multicolumn{2}{c}{\zvanilla} & \multicolumn{2}{c}{\yicesQS}\\
       & & \textsc{sat} & \textsc{unsat} & \textsc{sat} & \textsc{unsat} & \textsc{sat} & \textsc{unsat}\\
    \midrule
         \textbf{LIA}  & 416 & 150 & 266 & 150 & 266 & 107 & 102\\
         \textbf{LRA}  & 2\,419 & 795 & 1\,589 & 793 & 1\,595 & 808 & 1\,610 \\
\bottomrule\\
    \end{tabular}
    }
\caption{Instances solved within 20 minutes by different implementations. Benchmarks are quantified \textbf{LIA} and
  \textbf{LRA} formulas from SMT-LIB~\cite{BarFT-SMTLIB}.\label{tab:qsat1}}
\end{minipage}\hspace{.1in}
\begin{minipage}[t]{0.44\textwidth}
\centering
    \scalebox{.85}{
    \setlength{\tabcolsep}{1pt}
  \begin{tabular}{l r @{\hspace{0.1in}} a a @{\hspace{0.1in}} r r }
        \toprule
       \rowcolor{white}
       \multirow{2}{*}{Cat.} & \multirow{2}{*}{Count} & \multicolumn{2}{c}{\ztg} & \multicolumn{2}{c}{\zvanilla} \\
       & & \textsc{sat} & \textsc{unsat} & \textsc{sat} & \textsc{unsat} \\
    \midrule
\textbf{LIA-ADT} & 416 & 150 & 266 & 150 & 56  \\
         \textbf{LRA-ADT} & 2\,419 & 757 & 1\,415 & 196 & 964 \\
        \bottomrule\\
    \end{tabular}
    }
\caption{Instances solved within 60 seconds for our handcrafted benchmarks.\label{tab:qsat2}}
\end{minipage}
\vspace{-0.4in}
\end{table}

The second part of our evaluation shows the efficacy of \egraphmbp
for Arrays and ADTs~(\cref{alg:mbp}) in the context of CHC-solving. \zthree uses
both \qelite and MBP inside the CHC-solver \spacer~\cite{DBLP:conf/cav/KomuravelliGC14}. Therefore, we compare
\zvanilla and \ztg on CHC problems containing Arrays and ADTs. We use two sets of benchmarks to test out the efficacy of our MBP. The  benchmarks in the first set were generated for verification of
Solidity smart contracts~\cite{DBLP:conf/cav/AltBHS22}~(we exclude benchmarks with non-linear arithmetic, they are not supported by \spacer). These benchmarks have a very complex structure that nests ADTs and Arrays. Specifically, they contain both ADTs of Arrays, as well as Arrays of ADTs. This
makes them suitable to test our \egraphmbp. Row~1 of \cref{tab:mbp}
shows the number of instances solved by \zthree (\spacer) with and without \egraphmbp. \ztg solves $29$ instances more than \zvanilla. Even
though MBP is just one part of the overall \spacer algorithm, we see
that for these benchmarks, \egraphmbp makes a significant impact on \spacer.
Digging deeper, we find that many of these instances come from the category called \emph{abi}~(row~$2$ in \cref{tab:mbp}). \ztg solves all of
these benchmarks, while \zvanilla fails to solve~$20$ of them. We traced
the problem down to the MBP implementation in \zvanilla: it fails to eliminate all variables, causing runtime exception. In contrast, \egraphmbp eliminates all variables successfully, allowing \ztg to solve these benchmarks.

We also compare \ztg with \eld~\cite{DBLP:conf/fmcad/HojjatR18}, a state-of-the-art CHC-solver that is particularly
effective on these benchmarks. \ztg solves almost as many instances as
\eld. Furthermore, like \zvanilla, \ztg is orders of magnitude faster
than \eld. Finally, we compare the performance of \ztg on Array benchmarks
from the CHC competition~\cite{chccomp}. \ztg is competitive with \zvanilla, solving $2$ additional safe instances and almost as many unsafe instances as \zvanilla~(row~$3$ of
\cref{tab:mbp}). Both \ztg and \zvanilla solve quite a few instances~more~than~\eld.

\begin{table}[t]
    \centering
    \scalebox{0.9}{
    \begin{tabular}{l r@{\hspace{0.2in}} a a @{\hspace{0.2in}} r r@{\hspace{0.2in}}r r}
        \toprule
       \rowcolor{white}
             \multirow{2}{*}{Cat.} & \multirow{2}{*}{Count} & \multicolumn{2}{c}{\ztg} & \multicolumn{2}{c}{\zvanilla}& \multicolumn{2}{c}{\eld}\\
       & & \textsc{sat} & \textsc{unsat} & \textsc{sat} & \textsc{unsat} & \textsc{sat} & \textsc{unsat}\\
    \midrule
         \textbf{Solidity} & 3\,468 & 2\,324 & 1\,133 & 2\,314 & 1\,114 & \textbf{2\,329} & \textbf{1\,134}\\
         \textbf{ 	\rotatebox[origin=c]{180}{$\Lsh$} abi} & 127 & 19 & 108 & 19 & 88 & 19 &108\\
         \textbf{LIA-lin-Arrays} & 488 & \textbf{214} & 72 & 212 & \textbf{75} & 147 & 68\\
        \bottomrule\\
    \end{tabular}
    }
    \caption{Instances solved within 20 minutes by \ztg, \zvanilla, and \eld. Benchmarks are CHCs from \textbf{Solidity}~\cite{DBLP:conf/cav/AltBHS22} and  CHC competition~\cite{chccomp}. The \textbf{abi} benchmarks are a subset of \textbf{Solidity} benchmarks.}
    \label{tab:mbp}
    \vspace{-0.3in}
\end{table}

Our experiments show the effectiveness of our \qel and \egraphmbp in different settings inside the state-of-the-art SMT solver \zthree. While we maintain performance on quantified arithmetic benchmarks, we improve \zthree's QSAT algorithm on quantified benchmarks with ADTs. On verification tasks, \qel and \egraphmbp help \spacer solve 30 new instances, even though MBP is only a relatively small part of the overall \spacer algorithm.

\section{Conclusion}
\label{sec:conclusion}
Quantifier elimination, and its under-approximation, Model-Based
Projection are used by many SMT-based decision procedures, including
quantified SAT and Constrained Horn Clause solving. Traditionally,
these are implemented by a series of syntactic rules, operating
directly on the syntax of an input formula. In this paper, we argue
that these procedures should be implemented directly on the
egraph data-structure, already used by most SMT solvers. This results
in algorithms that better handle implicit equality reasoning and
result in easier to implement and faster procedures. We justify this
argument by implementing quantifier reduction and MBP in \zthree using
egraphs and show that the new implementation translates into
significant improvements to the target decision procedures. Thus, our
work provides both theoretical foundations for quantifier reduction
and practical contributions to \zthree SMT-solver.
 
\subsubsection*{Acknowledgment} The research leading to these results has received funding from the
European Research Council under the European Union's Horizon 2020 research and innovation programme (grant agreement No [759102-SVIS]). This research was partially supported by the Israeli Science Foundation (ISF) grant No. 1810/18. We acknowledge the support of the Natural Sciences and Engineering Research Council of Canada (NSERC), MathWorks Inc., and the Microsoft Research PhD Fellowship.

\clearpage
\bibliography{faqe}
\bibliographystyle{splncs04}
\ifthenelse{\noappendix = 0}{\clearpage
\appendix
\section{Proofs}
In this Appendix, we present proof for our claims.

\subsection{Admissibility of \tgreps\label{app:validityofrepr}}

Throughout the section, we assume that $\egtuple$ is an egraph, $n, m \in
\tgnodes$ are nodes in $G$, $\tgreps$ is a representative function for $G$, and
that $G_{\tgreps} = \langle \tgnodes, \tgedges_{\tgreps}\rangle$ is the graph in
the \cref{def:validrepr}. We define the \emph{execution trace} of an execution
of a procedure as the sequence of procedure calls made during the execution.

\begin{lemma}\label{lemma:ntt}
  $\toexpr(n_k, \tgreps)$ is in the execution trace of $\toexpr(n,
  \tgreps)$ iff there is a path from $n$ to $n_k$ in $G_{\tgreps}$.
\end{lemma}
\begin{proof}
  \emph{Only if direction.} By induction on the length of
  the path between $n$ and $n_k$ in $G_{\tgreps}$.

  \emph{Base case.} If the length is one, $n_k = \tgreps(n[1])$, the representative of the
  child of $n$. Since $n$ can not be of degree zero, $\toexpr(n,
  \tgreps)$ reaches line~\ref{ln:opfmlz3:reprchildren} and recurses on
  the representatives of all children of $n$. Since $n_k$ is one of
  them, $\toexpr(n_k, \tgreps)$ is in the execution trace of
  $\toexpr(n, \tgreps)$.

  \emph{Inductive step.} Assume that $n_j$ is reachable from $n$ in
  $G_{\tgreps}$ and that execution trace of $\toexpr(n, \tgreps)$
  contains $\toexpr(n_j, \tgreps)$. We prove that, if there exists an
  edge $(n_j, n_{j + 1})$ in $G_{\tgreps}$, execution trace of
  $\toexpr(n, \tgreps)$ contains $\toexpr(n_{j + 1}, \tgreps)$.
  Since there is an edge $(n_j, n_{j + 1})$, $n_{j + 1}$ is the
  representative of one of the children of $n_j$. Therefore,
  $\toexpr(n_{j}, \tgreps)$ reaches line~\ref{ln:opfmlz3:reprchildren}
  and calls $\toexpr(n_{j + 1}, \tgreps)$. Since the execution trace
  of $\toexpr(n, \tgreps)$ contains $\toexpr(n_j, \tgreps)$, it must
  also contain $\toexpr(n_{j + 1}, \tgreps)$.

  \emph{If direction.} We prove it by induction on the length of an
  execution trace.

  \emph{Base case.} The execution trace of $\toexpr(n, \tgreps)$ is
  of length 1 only if $n$ has exactly one child. Therefore, $n_k = n[1]$ is
  the representative of the only child of $n$. By
  \cref{def:validrepr}, $(n, n_k)\in E_{\tgreps}$.

  \emph{Inductive step.} 
  Our inductive
  hypothesis is that for all $j$ s.t. $1 < j \leq i$, $\toexpr(n_j,
  \tgreps)$ is in the execution trace and there is a path from $n$
  to $n_j$. We have to show that, if $\toexpr(n_{i + 1}, \tgreps)$ is
  in the execution trace, there is a path from $n$ to $n_{i +
    1}$. Since \cref{alg:eg-to-cube} recurses only on the
  representatives of children of a
  node~(line~\ref{ln:opfmlz3:reprchildren}), $\toexpr(n_{i + 1},
  \tgreps)$ is in the trace only if $n_{i + 1}$ is the representative
  of a child of a node $n_{j}$ where $1 \leq j \leq i$. By our
  inductive hypothesis, either $j = 1$ or there is already a path from
  $n$ to $n_j$. Since $n_{i + 1}$ is the representative of a child
  of $n_j$, by \cref{def:validrepr}, there is an edge from $n_j$ to
  $n_{i + 1}$. Hence, there is a path from $n$ to $n_{i + 1}$.  \qed
\end{proof}

\begin{lemma}
  If there is a path from $n$ to $n_k$ in $G_{\tgreps}$, then $n \neq
  \tgreps(n_k)$.
\end{lemma}

\begin{proof}
  Let the path from $n$ to $n_k$ be $n_1 = n, n_2,\ldots, n_k$ where
  $\forall 1 \leq j < k \cdot \exists i \cdot n_{j+1} =
  \tgreps(n_j[i])$. For contradiction, assume that $n =
  \tgreps(n_k)$. By \cref{def:validrepr}(a), each
  equivalence class has exactly one representative. Therefore
  $\tgreps(n_k) = \tgreps(n) = n$. The path $n, \ldots, n_{k -
    1}, n$ forms a cycle in $G_{\tgreps}$.  \qed
\end{proof}

\begin{corollary}\label{cor:ub-ntt}
    The maximum length of a path in $G_{\tgreps}$ is the number of
    equivalence classes in $G$.
\end{corollary}

\begin{lemma}\label{lemma:toformula-ub}
  Let $M$ and $d$ respectively be the number of equivalence classes
  and maximum out degree in $G$. The complexity of $\toexpr(n,
  \tgreps)$ is $\mathcal{O}(M^d)$.
\end{lemma}
\begin{proof}
  Let $T(n)$ denote the time to run $\toexpr(n, \tgreps)$. $T(n) =
  (\Sigma_{i = 1,\ldots, d} T(n[i])) + \mathcal{O}(1)$. This forms a
  tree with height $M$, the maximum length of the path in
  $G_{\tgreps}$. The number of leaf nodes is $M^d$.  \qed
\end{proof}

\Cref{lemma:toformula-ub} states that acyclicity of $G_{\tgreps}$ is a sufficient condition for
termination of $\toformula$. The following lemma states that it is also a necessary condition.

\begin{lemma}\label{lemma:toformula-lb}
  For a node $n\in\tgnodes$, let $p$ be the longest path starting from
  $n$ in $G_\tgreps$. The complexity of $\toexpr(n, \tgreps)$ is
  $\Omega(\mathit{length}(p))$.
\end{lemma}

\begin{proof}
  By \cref{lemma:ntt}, if there is a path from $n$ to
  some leaf $n_l$, $(n, \ldots, n_i,\ldots, n_l)$, all $n_i$ are on
  the execution trace of $\toexpr(n, \tgreps)$. \qed
\end{proof}
From \cref{lemma:toformula-lb}, it follows that if $G_\tgreps$ has a
cycle involving some node $n$, $\toexpr(n, \tgreps)$ does not
terminate, since there is no bound on the longest path, the
execution trace does not have a bound either.

\begin{lemma}\label{lma:toexpr}
  Given a formula $\varphi$ and $G = \egraph(\varphi)$ with an admissible representative function $\tgreps$. For every $n \in G$, $\varphi \models (\ntt(n) \eq \toexpr(n, \tgreps))$.
\end{lemma}

\begin{proof}
By \cref{cor:ub-ntt}, for every node $n$, $\toexpr(n, \tgreps)$ is defined. 

We prove by induction. For readability, we omit the $\tgreps$ parameter of \toexpr.
For the base case, if $n$ has no children $\ntt(n)$ is exactly $\toexpr(n)$.
Therefore, by condition (b) of admissibility of $\tgreps$, $\varphi \models \ntt(n) \eq \ntt(\tgreps(n))$. 

For the inductive case, let $f = L(n)$. For every child $n[i]$: \\
By hypothesis, $\varphi \models (\ntt(\tgreps(n[i])) \eq \toexpr(\tgreps(n[i])))$;  \\
by transitivity, $\varphi \models \ntt(n[i]) \eq \toexpr(\tgreps(n[i]))$; \\
by congruence, $\varphi \models (f(\ntt(n[i])) \eq f(\toexpr(\tgreps(n[i]))))$; \\ 
by definition, $f(\ntt(n[i])) = \ntt(n)$ and $f(\toexpr(n[i])) = \toexpr(n)$. \\
Therefore $\varphi \models (\ntt(n) \eq \toexpr(n))$.
\qed
\end{proof}
We are finally ready to prove \cref{lemma:repr-termination}. Recall:
\repeattheorem{admissible}

\begin{proof}
Termination in both directions follows from
\cref{lemma:toformula-ub,lemma:toformula-lb}.

Given a formula $\varphi$, and its egraph $G = \egraph(\varphi)$, and admissible $\tgreps$, let $\psi = G.\toformula(\tgreps, \emptyset)$ we show that $\formula(G,\psi)$.
We prove that $\varphi \leftrightarrow \psi$ is a valid formula. We do so in two steps.  (1) Let $\psi$ be of the form $(\ell_1, \ldots , \ell_n)$ then
$\varphi \models \ell_i$, and (2) let $\varphi$ be of the form $(\ell_1, \ldots , \ell_n)$ then $\psi \models \ell_i$.

Step (1), let $\psi$ be of the form $(\ell_1, \ldots , \ell_n)$ then
$\varphi \models \ell_i$. Each $\ell_i$ is of the form $\toexpr(r) \eq \toexpr(n)$ where $r = \tgreps(n)$.
By construction and completeness of egraphs, and conditions (a) and (b) of \cref{def:validrepr}, $\varphi \models \ntt(r) \eq \ntt(n)$.
By \cref{lma:toexpr}, $\varphi \models\ntt(r) \eq \toexpr(r)$ and $\varphi \models \ntt(n) \eq \toexpr(n)$. 
Then, by transitivity, $\varphi \models (\toexpr(r) \eq \toexpr(n))$.

Step (2), let $\varphi$ be of the form $(\ell_1, \ldots, \ell_n)$ then $\psi \models \ell_i$.  Each $\ell_i$ is of the form $\ntt(n) \eq \ntt(m)$ where $n$ and $m$ are nodes in $G$.
$\psi$ contains literals $\toexpr(\tgreps(n)) \eq \toexpr(n)$ and $\toexpr(\tgreps(m)) \eq \toexpr(m)$, then by transitivity we have $\psi \models \toexpr(n) \eq \toexpr(m)$. Without loss of generality, assume that $n = \tgreps(m)$, then $\psi \models \toexpr(n) \eq \toexpr(m)$. We prove that $\psi \models \toexpr(n) \eq \ntt(n)$, and the previous equality will follow by transitivity.
By induction. For the base case, $n$ has no children $\toexpr(n)$ is exactly $\ntt(n)$, therefore,  $\psi \models \toexpr(n) \eq \ntt(n)$.
For the inductive case, assume that for all children $n[i]$, $\psi \models \toexpr(n[i]) \eq \ntt(n[i])$. $\psi \models \toexpr(\tgreps(n)) \eq \toexpr(n)$, therefore $\psi \models \toexpr(\tgreps(n[i])) \eq \toexpr(n[i])$. Then, by congruence and transitivity $\psi \models f(\toexpr(\tgreps(n[i]))) \eq f(\ntt(n[i]))$, and by definition $\psi \models \toexpr(n) \eq \ntt(n)$.
\qed
\end{proof}

\subsection{Properties of \qel}\label{app:find-repr-proof}
In this section, we prove the properties of \qel. 
We start by proving that $\finddefs$ (\cref{alg:ground-valid-repr}) computes a maximally ground, admissible representative function~(\cref{lemma:finddefs}).

We extend the range of representative functions with a new symbol \emph{undefined}~($\undefrepr$) to represent that a node does not have a representative. We call such functions, \emph{partial} representative functions. For a partial representative function, we
say that $\tgreps(n)$ is \emph{defined} if $\tgreps(n) \neq
\undefrepr$.
In the following, let $G = \langle \tgnodes, E, L, \tgroot \rangle$ be
an egraph and $\tgreps = G.\finddefs(\vars)$ (\cref{alg:ground-valid-repr}). We use $r$,
$r_1$, $r_2$, $r_i$ to denote partial representative functions, the intermediate results computed during the execution of $\finddefs$.

\cref{alg:ground-valid-repr} begins with a partial representative function that
maps all nodes to $\undefrepr$. The function $\processQ(r, \acq)$ assigns
representatives to classes of nodes. For a class, the choice of representative
depends on the order in which $\acq$ is processed. That is, if $n$ and $m$
belong to the same class and $n$ appears before $m$ in $\acq$, $n$ will be the
representative of their class. However, for now, we are only interested on
whether a class has a representative after executing \processQ, not which
representative is chosen. The following lemma states that once a node is in the
$\acq$ list, it will get assigned a representative.

\begin{lemma}\label{lm:reprchosen}
  For any $n\in \tgnodes$, let $r_2 = \processQ(\tgreps, \acq)$ be such that $n
  \in \acq$ during the execution of $\processQ$. Then $r_2(n)\neq \undefrepr$
\end{lemma}
\begin{proof}
  Since $\processQ$ never removes representatives, if $\tgreps(n)\neq
  \undefrepr$, $r_2(n)\neq \undefrepr$. If $n \in \acq$, it is eventually
  processed. At this point, either $\tgreps(n)$ is already defined, or $n$ is
  chosen as the representative for its class. $\processQ$ terminates because in
  each iteration, it either decreases the number of classes without
  representative or, if the number of classes without representatives remains
  the same, it decreases the size of $\acq$.\qed
  \end{proof}

To prove that $\processQ(r, \acq)$ assigns a representative to a node
$n\in\tgnodes$, we analyze the descendants of $n$. If all children of
$n$ are in $\acq$, $\processQ(r, \acq)$ assigns a representative to
$n$. Furthermore, if $c \in \children(n)$ is not in $\acq$ but all
children of $c$ are in $\acq$, then, $\processQ(r, \acq)$ assigns
representative to $c$, adds $c$ to $\acq$, and then assigns
representative to $n$.
Thus, even if a descendant of $n$ is not in
$\acq$, it is enough that all the children of the descendant are in
$\acq$. Such sets are called \emph{frontiers} of a node. The set of
children of $n$ is one frontier of $n$. Given a frontier, replacing
any non-leaf node with all of its children gives us another
frontier. Formally, we define the set of all frontiers of $n$ as:

\begin{definition}[Frontiers of a node]
Given a graph $G = \langle \tgnodes, E\rangle$, the \emph{frontiers}
  of a node $n \in \tgnodes$, denoted $\frontiers(n)$, is a set of
  sets of nodes defined as
  \begin{multline*}
  \frontiers(n) = \{ \{\children(n)\} \} \cup \hphantom{a} \\
   \{ (\frnt \cup \{\children(m)\}) \setminus \{ m \} \mid \frnt \in \frontiers(n) \land m
  \in \frnt \land \degree(m) > 0 \}
  \end{multline*}
\end{definition}

\begin{lemma}\label{lma:processqfrontier}
 Let $r_2 = \processQ(r_1, \acq)$. For any node $n \in \tgnodes$ s.t.
 $\exists \frnt \in\frontiers(n) \cdot \frnt \subseteq \acq$, $r_2(n)
 \neq \undefrepr$.
\end{lemma}
\begin{proof}
  We prove it by induction on the frontiers of $n$. For the base case, $\frnt =
  \{\children(n)\}$. \finddefs assigns a representative to each $c \in \frnt$ at
  line~\ref{ln:chooserepr}, if they did not have one. After assigning a representative for the last child,
  the condition on line~\ref{ln:parentsel} is true for $n$. Therefore, $n$ is
  added to the $\acq$ and a representative for $n$ is
  chosen~(\cref{lm:reprchosen}). For the inductive step, assume that the lemma
  holds for a set $\frnt \in \frontiers(n)$. That is, if $\frnt \subseteq \acq$,
  and $r_2 = \processQ(r_1, \acq)$, $r_2(n)$ defined. We prove that the lemma
  holds for $\frnt' = (\frnt \cup \{\children(m)\}) \setminus \{ m \}$, where
  $m$ is a non-leaf node in $\frnt$. Since all nodes in $\frnt'$ are in $\acq$,
  all these nodes will eventually have representatives~(\cref{lm:reprchosen}).
  When the representative of the last child is chosen, $m$ is a node such that
  either it has a representative or all its children have representatives. In
  the second case, $m$ is added to $\acq$~(line~\ref{ln:parentsel}) and by
  \cref{lm:reprchosen}, its representative is eventually picked. Hence,
  $r_2(m)\neq \undefrepr$.\qed
\end{proof}
\begin{lemma}\label{lma:total-repr}
   For all $n\in \tgnodes$, $\tgreps(n)$ is defined.
\end{lemma}
\begin{proof}
  \cref{alg:ground-valid-repr} calls $\processQ$ with all leaf
  nodes. For any node $n$, there is a $\frnt \in \frontiers(n)$
  s.t. all nodes in $\frnt$ are leaf.  Therefore, by
  \cref{lma:processqfrontier} $\tgreps(n)$ is defined for all nodes.\qed
\end{proof}

\begin{definition}[Class frontiers of a node]
  \!\! Given an egraph $G\! =\! \langle \tgnodes,E,L,\tgroot\rangle$, the
  \emph{class frontiers} of a node $n \in \tgnodes$, denoted
  $\cfrontiers(n)$, is a set of sets of nodes defined as
  \begin{multline*}
  \cfrontiers(n) = \{ \{\children(n)\} \} \cup \hphantom{a}\\\{ (\frnt
  \cup \{c\}) \setminus \{ m \} \mid \frnt \in \cfrontiers(n) \land m
  \in \frnt \land c \in \classof(m) \} \cup \hphantom{a} \\ \{ (\frnt
  \cup \{\children(m)\}) \setminus \{ m \} \mid \frnt \in
  \cfrontiers(n) \land m \in \frnt \land \degree(m) > 0\}
  \end{multline*}
\end{definition}

\begin{lemma}\label{lma:cgrndfrontier}
    For any $n \in \tgnodes$ s.t. $\degree(n) > 0$, if $\cground(n)$, there
    exists a class frontier $\frnt \in \cfrontiers(n)$ s.t. each node in
    $\frnt$ is ground.
\end{lemma}
\begin{proof}
  $n$ is \cground if either (a) $\ntt(n)$ is ground or (b) if
  $\degree(n) > 0$ and for each $c \in \children(n)$, there exists a
  \cground node in $\classof(c)$. To construct $\frnt$, we follow this
  definition. Initially, let $\frnt = \{\children(n)\}$. For each $c$
  that have ground nodes in its class, we replace $c$ with the ground
  node. For each non-ground, \cground node $c'$ in $\frnt$, we replace
  $c'$ with $\children(c')$. We repeat this process until we get a
  class frontier with only ground nodes, the base case of the
  definition.\qed
\end{proof}

\begin{lemma}\label{lma:processqcfrontier}
 Let $r_2 = \processQ(r_1, \acq)$. For any node $n \in \tgnodes$ s.t
 $\exists \frnt \in\cfrontiers(n) \cdot \frnt \subseteq \acq$, $r_2(n)
 \neq \undefrepr$.
\end{lemma}
\begin{proof}
  We prove it by induction on the class frontiers of $n$. The base
  case, when $\frnt = \{\children(n)\}$, follows from
  \cref{lma:processqfrontier}. For the inductive step, assume that the
  lemma holds for a set $\frnt \in \cfrontiers(n)$. That is, if $\frnt
  \subseteq \acq$ and $r_2 = \processQ(r_1, \acq)$, $r_2(n)$ is
  defined. We prove that the lemma holds for both cases: (a) $\frnt' =
  (\frnt \cup \{\children(m)\}) \setminus \{ m \}$, where $m$ is a
  non-leaf node in $\frnt$, and (b) $\frnt'' = (\frnt \cup \{m'\})
  \setminus \{ m \}$, where $m\in\frnt$ and $m'\in\classof(m)$. Case
  (a) follows from the same arguments used in the proof of
  \cref{lma:processqfrontier}. Case (b) holds because both $m'$ and
  $m$ have the same representative, as they are in the same class.\qed
\end{proof}

\begin{lemma}\label{lemma:cground-repr}
 $\forall n\in \tgnodes \cdot \cground(n) \limp \cground(\tgreps(n))$
\end{lemma}
\begin{proof}
  For all ground nodes, all frontiers consists of only ground nodes. In
  particular, they have a frontier where all nodes are both ground and leaf.
  \cref{alg:ground-valid-repr} calls $\processQ$ with all ground, leaf nodes in
  $\acq$~(line~\ref{ln:grndprocessq}). Therefore, by
  \cref{lma:processqfrontier}, this call to $\processQ$ assigns representatives
  to all ground nodes in the egraph. At this point, all ground nodes either
  have representatives or are in $\acq$. Therefore, by \cref{lma:cgrndfrontier},
  each \cground node in the egraph has a frontier in $\acq$. Hence, by
  \cref{lma:processqcfrontier}, $\processQ$ assigns representatives for all
  \cground nodes.\qed
\end{proof}

We extend the definition of admissibility of representative functions to
partial representative functions as follows.
\begin{definition}[Admissible partial representative functions]\label{def:validpartialrepr}
Given an egraph $G = \langle \tgnodes, E, L, \tgroot\rangle$, we say
that a function $\tgreps : \tgnodes \rightarrow \tgnodes \cup
\{\undefrepr\}$, is an \emph{admissible partial representative function} for $G$ if:
\begin{itemize}\itemsep=0pt
\item[(a)]
$\forall n \in \tgnodes\cdot (\tgreps(n) \neq \undefrepr) \limp (\forall n' \in \classof(n) \iff
\tgreps(n) = \tgreps(n'))$, 
\item[(b)] the graph $\langle \tgnodes,
E_{\tgreps}\rangle$ is acyclic where $E_{\tgreps} = \{ (n, \tgreps(c))
\mid n \in \tgnodes, c \in \children(n), \tgreps(c)\neq \undefrepr \}$, and
\item[(c)] $\forall n, n'\in \tgnodes\cdot n = \tgreps(n') \limp
\forall c \in \children(n)\cdot \tgreps(c) \neq \undefrepr$.
\end{itemize}
\end{definition}
The first two conditions state that $\tgreps$ is admissible for all nodes
that it defines. The third condition states that, for all nodes in the
range of $\tgreps$, the representatives for their children are also
defined. This allows building admissible partial representative functions based on already admissible partial representative functions.

\begin{lemma}\label{lma:validpartreprinv}
  Let $r$ be an admissible partial representative function. Let $n$ be a
  node s.t. $r(n) = \undefrepr$ and $\forall c \in \children(n)\cdot r(c)
  \neq \undefrepr$. Let $r_1$ be a partial representative function defined as
  $r_1(n') = ite(n' \in \classof(n), n, r(n'))$. $r_1$ is an \emph{admissible}
  partial representative function.
\end{lemma}
\begin{proof}
 By condition (c) of \cref{def:validpartialrepr}, for any child $c$ of
 $n$ and for any node $m$ reachable from $c$ in the graph $\langle N,
 E_{r}\rangle$, $r(m)$ is defined. Therefore, $r(m)\not \in
 \classof(n)$ as $r(n)$ is undefined. Therefore, choosing $n$ as
 representative of its equivalence class does not introduce any
 cycles.\qed
\end{proof}
\begin{lemma}\label{lemma:finddefs}
  Representative functions $\tgreps$ computed by \finddefs are maximally ground and satisfy \cref{def:validrepr}.
\end{lemma}
\begin{proof}  
  \cref{alg:ground-valid-repr} chooses a node $n$ as class representative
  either if $n$ is a leaf or if the representatives of
  all children of $n$ have been chosen. Since choosing leaves as
  representatives does not introduce cycles, the first case preserves
  admissibility of partial representative
  functions. \cref{lma:validpartreprinv} shows that the second case
  also preserves admissibility of partial representative functions.
  By \cref{lma:total-repr}, $\tgreps$ is a total function.
  Ground maximality is stated in~\cref{lemma:cground-repr}.
  \qed
\end{proof}

Finally, we prove that finding a maximally ground, admissible
representative function is sufficient to find ground definitions for any variables that have it. Recall \cref{thm:gr-defs-new}:
\repeattheorem{thm:gr-defs-new}
\begin{proof}
  From completeness of egraphs, it follows that if $\varphi \models
  \ntt(n)\eq t$ for some ground term $t$, $\classof(n)$ is ground. By definition of maximally ground, for all such nodes,
  $\tgreps(n)$ is \cground. Next, we prove that for all representative
  nodes $r\in \mathit{range}(\tgreps)$ that are c-ground, $\toexpr(r)$ is
  ground.
  
  We prove this by induction on $G_\tgreps$. For the base case, all leaves $l \in \mathit{range}(\tgreps)$ s.t. $\cground(l)$,
  $\toexpr(l)$ is ground. By definition of \cground, leaves are \cground if they are ground. Therefore, $\toexpr(l) =
  \ntt(l)$ is ground. For the inductive case, we assume that, for a
  c-ground representative node $n$, all its children $n[i]$ have the property that $\toexpr(\tgreps(n[i]))$ is
  ground. Therefore, for $f = L(n)$, and $\mathit{args}$ s.t. $\mathit{args}[i] = \toexpr(\tgreps(n[i]))$, $\toexpr(n) = f(\mathit{args})$ is also ground.
  \qed
\end{proof}

\noindent
We are now ready to prove relative completeness of \qel. Recall:

\repeattheorem{qelsummary}

\begin{proof}
For condition (1), variables that have a ground representative are excluded from \core, therefore, they are excluded in \toformula, and are successfully eliminated.

For condition (2), \toformula calls \toexpr only on nodes in $NS$, this triggers a call to \toexpr only for all representatives of the children, i.e., only the ones reachable in $G_\tgreps$. Therefore, if the node is not reachable in $G_\tgreps$, it does not appear in the output formula. \qed
\end{proof}
 
\section{Implementation of MBP rules within egraphs}\label{sec:mbprules}
\newcommand{\hasvars}{\mathit{has\_vars}}
\newcommand{\hasvar}{\mathit{has\_var}}
\newcommand{\isvar}{\mathit{is\_var}}
\newcommand{\peq}{=^p}
In this section, we give egraph implementations of most of the rules for Array
MBP~\cite{DBLP:conf/fmcad/KomuravelliBGM15} and ADT
MBP~\cite{DBLP:conf/lpar/BjornerJ15}. $\peq$ is the partial equality predicate defined in \cite{DBLP:conf/fmcad/KomuravelliBGM15}.

\begin{figure}
  \centering

  \scalebox{0.9}{
\begin{minipage}[t]{0.45\textwidth}
\begin{tabular}{l}
$\inferrule*[lab=\textit{PartialEq}]{\varphi[e_1\eq e_2]}{\varphi[e_1\peq_{\emptyset} e_2]}$\\
\end{tabular}
\end{minipage}\begin{minipage}{0.54\textwidth}
  $\textit{PartialEq}$
  \begin{algorithmic}[1]
    \Function{\checkr}{$t$}
    \State \Return $t = e_1 \eq e_2 \land \hasvars(t)$
    \EndFunction
  \Function{\applyr}{$t, M, G$}
  \State $G.\textit{assert}(e_1 \peq_{\emptyset} e_2)$
  \EndFunction
  \end{algorithmic}
\end{minipage}
  }

  \scalebox{0.9}{
\begin{minipage}[t]{0.45\textwidth}
  \begin{tabular}{l}
 $\inferrule*[lab=\textit{ElimEq}, right=$\neg \hasvar (e\, v)$]{\exists v\cdot v \peq_{\Vec{i}} e\land \varphi}{\varphi [a \rightarrow \arwrite(e, \Vec{i}, \Vec{d})]}$
\end{tabular}
\end{minipage}
\begin{minipage}
{0.54\textwidth}
  $\textit{ElimEq}$
  \begin{algorithmic}[1]
    \Function{\checkr}{$t$}
    \State \Return $t =  v \peq_{\Vec{i}} e \land{}$
    \Statex \hfill $\isvar(v) \land \neg\hasvar(e, v)$
    \EndFunction
    \Function{\applyr}{$t, M, G$}
    \State $\Vec{d} := \textit{newVars}()$
    \State $G.\textit{assert}(v \eq \arwrite(e, \Vec{i}, \Vec{d}))$
    \State $G.\textit{assert}(t \eq \top)$
    \State $\forall 0 \leq j < len(\Vec{i}) \cdot M[\Vec{d}_j] := M[\textit{select}(v, \Vec{i}_j)]$
  \EndFunction
  \end{algorithmic}
\end{minipage}
  }

   \scalebox{0.9}{
  \begin{minipage}{0.45\textwidth}
    \centering
\begin{tabular}{l}
  $\inferrule*[lab=\textsc{Ackermann1}, right=$M \models t_1 \eq t_2$]{\varphi\land \arread(v, t_1) \eq s_1\land \hphantom{a}\\read(v, t_2) \eq s_2}{\varphi \land s_1 \eq s_2 \land t_1 \eq t_2}$\\ \\
$\inferrule*[lab=\textsc{Ackermann2}, right=$M\models i \deq j$]{\varphi \land \arread(v, t_1) \eq s_1\land\hphantom{a}\\ read(v, t_2) \eq s_2}{\varphi \land t_1 \deq t_2}$
\end{tabular}
\end{minipage}
\begin{minipage}{0.54\textwidth}
  \small
  $\textit{Ackermann}$
  \begin{algorithmic}[1]
    \Function{\checkr}{$\langle t_1, t_2\rangle$}
    \State \Return $t_1 = \arread(v, e_1) \land{}$\\
    \hfill $t_2 = \arread(v, e_2) \land e_1 \neq e_2$
    \EndFunction
  \Function{\applyr}{$t_1, t_2, M, G$}
  \If {$M \models e_1 \eq e_2$}
  \State $G.\textit{assert}(e_2 \eq e_2)$
  \Else
  \State $G.\textit{assert}(e_1 \deq e_2)$
  \EndIf
  \EndFunction
  \end{algorithmic}
\end{minipage}
}
\caption{Implementation of additional Array rules from~\cite{DBLP:conf/fmcad/KomuravelliBGM15}. The \textsc{Ackermann} rule to eliminate $\arread$
  terms and its implementation in egraphs. This rule is applied
  to each pair of $\arread(v, t)$ terms.\label{fig:ackerman} \label{fig:array-extra}}
\end{figure}

 }{}
\end{document}